\documentclass[fullpage]{article}
\usepackage{CJK}
\usepackage[affil-it]{authblk}
\usepackage{amsmath}
\usepackage{mathrsfs}
\usepackage{xcolor}
\usepackage{amssymb}
\usepackage{amsthm}
\usepackage{MnSymbol}
\usepackage{bm}
\usepackage{comment}
\usepackage{hyperref}
\usepackage{soul}

\usepackage{tikz}
\usetikzlibrary{positioning}
\usepackage{amsfonts}
\usepackage{amsmath}
\usepackage{amssymb}
\usepackage{amsthm}
\usepackage{mathtools}
\usepackage{enumerate}
\usepackage{bm}
\usepackage{dsfont}
\usepackage{graphicx}
\usepackage{nicefrac}
\usepackage{color}
\usepackage[all]{xy}
\usepackage{hyperref}
\usepackage{bbm}
\usepackage{tikz}
\usepackage{tikz-cd}
\usepackage{cancel}
\usepackage{xcolor}
\usepackage{booktabs}
\usepackage{multirow}
\usepackage{subcaption}
\usetikzlibrary{shapes.geometric}
\usetikzlibrary{backgrounds,fit,decorations.pathreplacing} 
\usetikzlibrary{circuits, calc}
\usepackage{stackengine}
\usepackage{adjustbox}
\usepackage{placeins}
\usepackage[bottom]{footmisc}

\newtheorem{theorem}{Theorem}

\newtheorem{corollary}{Corollary}

\newtheorem*{question*}{Question}
\newtheorem{definition}{Definition}
\theoremstyle{remark}
\newtheorem{example}{Example}

\newcommand{\mf}[1]{\textcolor{black}{#1}}

\definecolor{nblue}{rgb}{0.3,0.3,1.0}
\definecolor{ngreen}{rgb}{0.2,0.7,0.2}
\definecolor{nred}{rgb}{0.9,0.1,0}
\definecolor{red2}{rgb}{0.6,0.2,0.2}
\definecolor{npurple}{rgb}{0.8,0.2,0.8}
\definecolor{golden}{rgb}{0.8,0.6,0.1}
\definecolor{nsilver}{rgb}{0.3,0.4,0.5}
\definecolor{nbrown}{rgb}{0.8,0.4,0.15}
\definecolor{nrose}{rgb}{0.7,0,0.35}
\definecolor{nviol}{rgb}{0.5,0,1.0}
\definecolor{nazur}{rgb}{0,0.35,0.7}
\definecolor{nchart}{rgb}{0.2,0.4,0}
\definecolor{nbrick}{rgb}{0.55,0.25,0.15}

\newcommand{\blk}{\color{black}}
\newcommand{\red}{\color{black}}
\newcommand{\pur}{\color{black}}

\newcommand{\tr}[1]{\ensuremath{\mathrm{Tr}}#1}
\newcommand\ra{\rightarrow}

\newcommand\mc[1]{\mathcal{#1}}
\newcommand\cH{\mc{H}}
\newcommand\cM{\mc{M}}
\newcommand\Ain{{\cH_{A^{\mathrm{in}}}}}
\newcommand\Aout{{\cH_{A^{\mathrm{out}}}}}
\newcommand\Bin{{\cH_{B^{\mathrm{in}}}}}

\newcommand\Ai{{A^{\mathrm{in}}}}
\newcommand\Ao{{A^{\mathrm{out}}}}
\newcommand\Bi{{B^{\mathrm{in}}}}

\newcommand\out{\mathrm{out}}
\newcommand\ins{\mathrm{in}}
\newcommand\Pa{\mathrm{Pa}}

\newcommand\Ch{\mathrm{Ch}}

\newcommand{\bA}{\mathbf{A}}
\newcommand{\ba}{\mathbf{a}}
\newcommand{\bB}{\mathbf{B}}
\newcommand{\bb}{\mathbf{b}}
\newcommand{\bC}{\mathbf{C}}
\newcommand{\bc}{\mathbf{c}}
\newcommand{\bE}{\mathbf{E}}
\newcommand{\be}{\mathbf{e}}
\newcommand{\bF}{\mathbf{F}}

\newcommand{\bT}{\mathbf{T}}
\newcommand{\bt}{\mathbf{t}}
\newcommand{\bU}{\mathbf{U}}
\newcommand{\bu}{\mathbf{u}}
\newcommand{\bV}{\mathbf{V}}
\newcommand{\bv}{\mathbf{v}}
\newcommand{\bX}{\mathbf{X}}
\newcommand{\bx}{\mathbf{x}}
\newcommand{\bY}{\mathbf{Y}}
\newcommand{\by}{\mathbf{y}}
\newcommand{\bZ}{\mathbf{Z}}
\newcommand{\bz}{\mathbf{z}}

\newcommand{\M}{M} 
\newcommand{\Q}{M_Q} 

\newcommand{\bL}{\mathbf{\Lambda}}

\begin{document}

\vspace{-2cm}
\title{\textbf{A Semantics for Counterfactuals\\ in Quantum Causal Models}}
\author{Ardra Kooderi Suresh\thanks{a.kooderisuresh@griffith.edu.au} $^1$, Markus Frembs\thanks{m.frembs@griffith.edu.au} $^1$, Eric G. Cavalcanti\thanks{e.cavalcanti@griffith.edu.au} $^1$}
\affil{$^1$\normalsize{Centre for Quantum Dynamics, Griffith University,\\ Yugambeh Country, Gold Coast, QLD 4222, Australia}}

\date{}

\maketitle

\vspace{-1cm}
	
\begin{abstract}
We introduce a formalism for the evaluation of counterfactual queries in the framework of quantum causal models, generalising Pearl's semantics for counterfactuals in classical causal models~\cite{Pearl2000}, thus completing the last rung in the quantum analogue of Pearl’s “ladder of causation”. To this end, we define a suitable extension of Pearl's notion of a `classical structural causal model', which we denote analogously by \emph{`quantum structural causal model'}, and a corresponding extension of Pearl's three-step procedure of abduction, action, and prediction. We show that every classical (probabilistic) structural causal model can be extended to a quantum structural causal model, and prove that counterfactual queries that can be formulated within a classical structural causal model agree with their corresponding queries in the quantum extension -- but the latter is more expressive. Counterfactuals in quantum causal models come in different forms: we distinguish between \emph{active} and \emph{passive} counterfactual queries, depending on whether or not an intervention is to be performed in the action step. This is in contrast to the classical case, where counterfactuals are always interpreted in the active sense. Another distinctive feature of our formalism is that it break the connection between causal and counterfactual dependence that exists in the classical case: quantum counterfactuals allow for \emph{counterfactual dependence without causal dependence}. This distinction between classical and quantum causal models may shed light on how the latter can reproduce quantum correlations that violate Bell inequalities while being faithful to the relativistic causal structure.
\end{abstract}

\vspace{.3cm}

\section{Introduction}

The world of alternative possibilities has been pondered upon and analyzed routinely, in many fields of study including but not limited to social \cite{fearon1996causes} and public policy \cite{loi2012note}, psychiatry \cite{tagini2021counterfactual}, economy \cite{ravallion2021poverty}, weather and climate change \cite{woo2021counterfactual}, artificial intelligence \cite{holtman2021counterfactual}, philosophy and causality \cite{hoerl2011understanding, collins2004causation}. For example, questions involving counterfactuals can have important social and legal implications, such as ``Given that the patient has died after treatment, would they have survived had they been given a different treatment?''\pur , or ``how many lives could the US have saved had it authorized booster vaccines sooner?"\cite{black2023covid}.\blk

The status of counterfactual questions also figures centrally in debates about quantum mechanics \cite{vaidman2009counterfactuals}, where results such as Bell's theorem \cite{Bell1964} and the Kochen-Specker theorem~\cite{KochenSpecker1967} have been interpreted as requiring the abandonment of ``counterfactual definiteness''\cite{skyrms1982counterfactual}, encapsulated in Peres' famous dictum ``unperformed experiments have no results'' \cite{Peres1978}. \red Could this assertion be used by a lawyer in an argument to dismiss a medical malpractice lawsuit as meaningless? \blk Presumably not. Dismissing \emph{all} counterfactual questions as meaningless due to quantum theory \red thus \blk seems too strong. Here, we seek to delineate what \red counterfactual questions involving quantum systems \blk can be unambiguously answered when unambiguously formulated, and to provide some direction for resolving the ambiguity that is inherent in counterfactual questions that are not so carefully constructed.

The semantics of counterfactuals has a controversial history. In one of the early accounts, David Lewis \cite{lewis1973counterfactuals} proposed to evaluate counterfactuals via a similarity analysis of possible worlds, where ``a counterfactual ‘If it were that A, then it would be that C’ is (non-vacuously) true if and only if some (accessible) world where both A and C are true is more similar to our actual world, overall, than is any world where A is true but C is false'' \cite{lewis1973counterfactuals}. This analysis is inevitably vague, as it requires an account of ``similarity'' among possible worlds, which Lewis attempts to resolve via a system of priorities. The goal is to identify \emph{closest worlds} as possible worlds in which things are kept more or less the same as in our actual world, except for some `minimal changes', required to make the antecedent of a given counterfactual true.

A recent approach, due to Judea Pearl, proposes to define counterfactuals in terms of a sufficiently well-specified causal model for a given situation, denoted by a (classical) \emph{structural causal model} \cite{Pearl2000}. In Pearl's approach, the `minimal changes' required to make the antecedent of a counterfactual true are conceptualised in terms of an \emph{intervention}, which breaks the causal connections into the variable being intervened upon while fixing it to the required counterfactual value. Structural causal models feature at the top of a hierarchy of progressively sophisticated models that can answer progressively sophisticated questions, which Pearl has dubbed the ``ladder of causation''~\cite{pearl2018book} (see Fig.~\ref{fig:ladder}).

As is well known, however, the classical causal model-framework of Pearl fails to reproduce quantum correlations while maintaining faithfulness to the relativistic causal structure---as vividly expressed by Bell's theorem~\cite{Bell1964} and recent `fine-tuning theorems'~\cite{wood2015lesson,cavalcanti2018classical,pearl2019classical}. The program of \emph{quantum causal models}~\cite{LeiferSpekkens2013, cavalcanti2014modifications, Henson2014, Chaves2015, Pienaar2015, costa2016quantum,  AllenEtAl2017, ShrapnelCostaMilburn2018, LorenzOreshkovBarrett2019} aims to resolve this tension by extending the classical causal model framework, while maintaining compatibility with relativistic causality. One of the aims of our work is to complete the last rung in the quantum analogue of Pearl's ``ladder of causation'', by proposing a framework to answer counterfactual queries in quantum causal models.

\pur
Pearl argues that the levels of interventions and counterfactuals are particularly important for human intelligence and understanding, as they are crucial for our internal modeling of the world and of the effects of our actions. In contrast, he argues that current artificial intelligence (AI) models ---however impressive--- are still restricted to level 1 of his causal hierarchy \cite{pearl2018book}. Another motivation of our extension of Pearl’s analysis to the framework of quantum causal models is thus its potential applications for quantum AI. 
\blk

A key distinction from the classical case is that, due to the indeterminism inherent in quantum causal models, counterfactual queries do not always have truth values (unlike in Lewis' and Pearl's accounts). Another difference is that an intervention is not always required in order to make the antecedent of a counterfactual true. This leads to a richer semantics for counterfactuals in the quantum case, which contains Pearl's classical structural causal model as a special case, as we show.

Finally, an important distinction regards the connection between counterfactual dependence and causal dependence. In Pearl's account, counterfactual dependence requires causal dependence. Similarly, Lewis \cite{Lewis1973causation} proposed an analysis of causal dependence based on his own notion of counterfactual dependence. In contrast, in quantum causal models there can be counterfactual dependence among events without causal dependence. This fact sheds new light on the nature of the compatibility with relativistic causality that is offered by quantum causal models. It can be thought of as a clarification and generalisation of Shimony's notion of ``passion at a distance''~\cite{Shimony1984}.

The rest of the paper is organised as follows. In Sec.~\ref{sec: classical causal models}, we review the basic ingredients to Pearl's ladder of causation (see Fig.~\ref{fig:ladder}), as well as his three-step procedure for evaluating counterfactuals based on the notion of (classical) structural causal models. In Sec.~\ref{sec: quantum violations}, we highlight the issues in accommodating quantum theory within this framework, in the light of Bell's theorem and the assumption of ``no fine-tuning'' \cite{wood2015lesson, cavalcanti2018classical}. The framework of quantum causal models aims to resolve this discrepancy. We introduce some key notions and notation of the latter in Sec.~\ref{sec: quantum causal models}, which will set the stage for our definition of quantum counterfactuals and their semantics based on a novel notion of quantum structural causal models in Sec.~\ref{sec: counterfactuals in QSM}. In Sec.~\ref{sec: quantum extension}, we show that Pearl's classical formalism of counterfactuals naturally embeds into our framework; conversely, in Sec.~\ref{sec:analysis} we elaborate on how our framework generalizes Pearl's formalism, by distinguishing passive from active counterfactuals in quantum causal models. This results in a difference between causal and counterfactual dependence in quantum causal models, which pinpoints a remarkable departure from classical counterfactual reasoning. We discuss this using the pertinent example of the Bell scenario in Sec.~\ref{sec: Bell scenario}. \pur In Sec.~\ref{sec: CFD}, we briefly review the debate surrounding the notion of counterfactual definiteness in quantum mechanics, and how although it fails in our formalism, this is not the particularly distinctive feature of quantum counterfactuals, and this rejection by itself cannot resolve Bell's theorem. \blk Sec.~\ref{sec: quantum Bayes} reflects on some of the key assumptions to our notion of quantum counterfactuals in the context of recent developments on quantum statistical inference. Sec.~\ref{sec: conclusion} concludes with a brief summary and discussion of the results, and questions for further work.
\begin{figure}
\centering
	\includegraphics[width=1\linewidth]{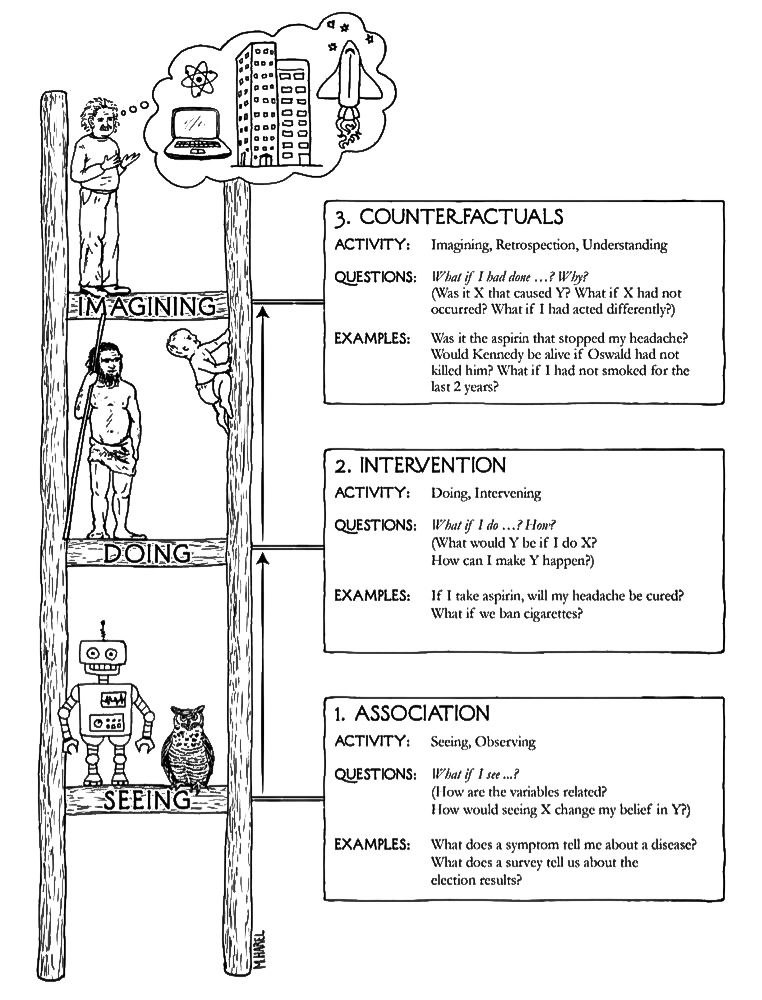}
	\caption{A depiction of ``The Ladder of Causation". [Republished with permission from Ref.~\cite{pearl2018book}.]}
	\label{fig:ladder}
\end{figure}

\section{The Classical Causal Model Framework}\label{sec: classical causal models}

This section contains the minimal background on classical causal models and the evaluation of counterfactuals required for the generalization to the quantum case in Sec.~\ref{sec: counterfactuals in QSM}. We will review a small fraction of the framework outlined in much more detail in Ref.~\cite{Pearl2000}; readers familiar with the latter may readily skip this section.

In his book on causality \cite{Pearl2000}, Judea Pearl identifies a hierarchy of progressively sophisticated models that are capable of answering progressively sophisticated causal queries. This hierarchy is often depicted as the three-rung `Ladder of Causation' (see Fig. \ref{fig:ladder}).

At the bottom of the ladder is the level of \emph{association} (`level 1'), related to observations and statistical relations. It answers questions such as ``how would seeing $X$ change my belief in $Y$?" The second rung is the level of \emph{intervention} (`level 2'), which considers questions such as “If I take aspirin, will my headache be cured”? The final rung in the ladder of causation is the level of \emph{counterfactuals} (`level 3'), associated with activities such as imagining, retrospecting, and understanding. It considers questions such as ``Was it the aspirin that stopped my headache?”, ``Had I not taken the aspirin, would my headache not have been cured?" etc. In other words, counterfactuals deal with `why'-questions. Formally, levels 1, 2 and 3 are related to \emph{Bayesian networks}, \emph{causal Bayesian networks}, and \emph{structural causal models}, respectively. We will formally define these in the coming subsections. 

\subsection{Level 1 - Bayesian networks}

In Pearl's framework, level 1 of the ladder of causation (Fig. \ref{fig:ladder}) is the level of association, which encodes statistical data in the form of a probability distribution $P(\bv) = P(v_1,\cdots,v_n)$ over random variables $\bV = \{V_i\}_{i=1}^n$.\footnote{Throughout, we will use boldface notation to indicate tuples of variables.} The latter are assumed to take values in a finite set, whose elements are denoted by the corresponding lowercase $v_i$. The proposition `$V_i=v_i$' represents an \textit{event} where the random variable $V_i$ takes the value $v_i$, and $P(v_i) := P(V_i=v_i) $ denotes the probability that this event occurs. 

Statistical independence conditions in a probability distribution can be conveniently  represented graphically using directed acyclic graphs (DAGs), which in this context are also known as \emph{Bayesian networks}. The nodes in a Bayesian network $G$ represent the random variables $\bV = \{V_i\}_{i=1}^n$, while arrows (`$V_j \ra V_k$') in $G$ impose a `kinship' relation: we call $\Pa(V_i) = \Pa_i := \{V_j \in \bV \mid (V_j \ra V_i) \in G\}$ the ``parents'' and $\mathrm{Ch}(V_i) = \Ch_i := \{V_j \in \bV \mid (V_i \ra V_j) \in G\}$ the ``children'' of the node $V_i$. For example, in Fig.~\ref{DAG_basic}, $V_1$ is the parent node of $V_2$ and $V_3$; $V_4$ is a child node of $V_2$, $V_3$ and $V_6$. 

\begin{figure}[b]
     \centering
		\begin{tikzpicture}[node distance={24mm}, thick, main/.style = {draw, circle}] 
			\node[main] (1) {$V_1$}; 
			\node[main] (2) [above right of=1] {$V_2$}; 
			\node[main] (3) [below right of=1] {$V_3$}; 
			\node[main] (4) [above right of=3] {$V_4$}; 
			\node[main] (5) [above right of=4] {$V_5$}; 
			\node[main] (6) [below right of=4] {$V_6$}; 
			\draw[->] (1) -- (2); 
			\draw[->] (1) -- (3);  
			\draw[->] (2) -- (4); 
			\draw[->] (3) -- (4); 
			\draw[->] (4) -- (5); 
			\draw[->] (6) -- (4);
		\end{tikzpicture} 
		\caption{A directed acyclic graph (DAG) with nodes $\bV = \{V_1,\cdots,V_6\}$ representing random variables, and arrows representing (causal) statistical dependencies.} \label{DAG_basic}
\end{figure}
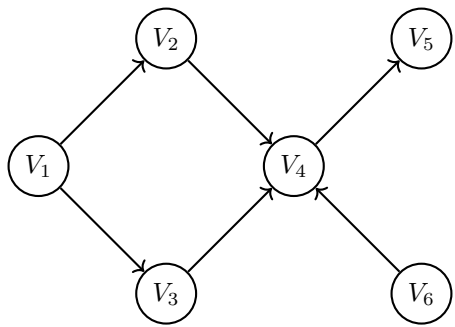

\begin{definition}[Classical Markov condition]\label{def: cmc}
		A joint probability distribution $P(\bv) = P(v_1, \cdots, v_n)$ is said to be \textit{Markov relative to a DAG $G$ with nodes $\bV = \{V_i\}_{i=1}^n$} if and only if there exist conditional probability distributions $P(v_i|pa_i)$ for each $V_i \in \bV$ such that,
		\begin{equation}\label{eq: Markov condition}
			P(\bv) = \prod_{i=1}^{n} P(v_i|pa_i)\; .
		\end{equation}
\end{definition}

In general, a probability distribution may be Markov relative to many Bayesian networks, corresponding to different ways it can be decomposed into conditional distributions. Moreover, a Bayesian network will have many distributions which are Markov with respect to it. Note that at this level (level 1), the DAG $G$ representing a Bayesian network does not carry causal meaning, but is merely a convenient representation of statistical conditional independences.

\subsection{Level 2 - Causal Bayesian networks and classical causal models}\label{sec: causal Bayesian networks}

At level 2 of the hierarchy are \emph{causal (Bayesian) networks}. In contrast to Bayesian networks, the arrows between nodes $\bV = \{V_i\}^n_{i=1}$ in a causal Bayesian network \emph{do} encode causal relationships. In particular, the parents $\Pa(V_i)$ of a node $V_i$ are now interpreted as \emph{direct causes} of $V_i$. Moreover, a causal network is an oracle for interventions. The effect of an intervention is modeled as a ``mini-surgery” in the graph that cuts all incoming arrows into the node being intervened upon and sets it to a specified value. We define the \emph{do-intervention} $\mathrm{do}(\bX = \bx)$ on a subset of nodes $\bX \subset \bV$ as the submodel $\langle G_\bx,P_\bx\rangle$, where $G_\bx$ is the modified DAG with the same nodes as $G$, but with all incoming arrows $V_j \ra V_i$ for $V_i \in \bX$ removed from $G$, and where $P_\bx$ arises from $P$ by setting the values at $\bX$ to $\bx$. More precisely, letting $\bV_\bx = \bV \backslash \bX$
\begin{equation}\label{eq: do-intervention}
    P_\bx(\bv) = \prod_{V_i \in \bX} \delta_{v_i,x_i} \prod_{V_i \in \bV_\bx} P(v_i|pa_i)\;.
\end{equation}

\begin{definition}[Classical Causal Model]\label{def: CCM}
    A \emph{classical causal model} is a pair $\langle G, P \rangle$, consisting of a directed acyclic graph $G$ \mf{with nodes $\bV = \{V_i\}_{i=1}^n$}, a probability distribution $P$ that is Markov with respect to $G$, according to Def.~\ref{def: cmc} \mf{and all its submodels $\langle G_\bx,P_\bx\rangle$, arising from do-interventions $\mathrm{do}(\bX = \bx)$ with $\bX \subset \bV$.}
\end{definition}

\mf{For example, if we perform a do-intervention $\mathrm{do}(V_2=v_2)$ on the classical causal model $\langle G,P\rangle$ with DAG $G$ in Fig.~\ref{DAG_basic}, then $G_{v_2}$ is the DAG shown in Fig.~\ref{DAG_intervened}, and the truncated factorization formula for the remaining variables reads
\begin{equation}
		P_{v_2}(v_1,v_3,v_4,v_5,v_6) = P(v_1) P(v_3|v_1) P(v_4|V_2=v_2,v_3,v_6) P(v_5|v_4) P(v_6)\;.
\end{equation}}

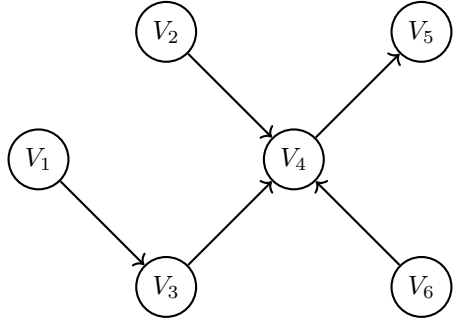
\begin{figure}
         \centering
		\begin{tikzpicture}[node distance={24mm}, thick, main/.style = {draw, circle}] 
			\node[main] (1) {$V_1$}; 
			\node[main] (2) [above right of=1] {$V_2$}; 
			\node[main] (3) [below right of=1] {$V_3$}; 
			\node[main] (4) [above right of=3] {$V_4$}; 
			\node[main] (5) [above right of=4] {$V_5$}; 
			\node[main] (6) [below right of=4] {$V_6$}; 
			\draw[->] (1) -- (3);  
			\draw[->] (2) -- (4); 
			\draw[->] (3) -- (4); 
			\draw[->] (4) -- (5); 
			\draw[->] (6) -- (4);
		\end{tikzpicture} 
		\caption{The directed acyclic graph from Fig.~\ref{DAG_basic} after a do-intervention on node $V_2$. The effect of this do-intervention is graphically represented by removing all the arrows into $V_2$.} \label{DAG_intervened}
\end{figure}

\subsection{Level 3 - Structural causal models and the evaluation of counterfactuals}

At level 3 of the hierarchy are (classical) \emph{structural causal models}. Such models consist of a set of nodes $(\bV,\bU)$, distinguished into \emph{endogenous} variables $\bV$ and \emph{exogenous} variables $\bU$, together with a set of functions $\bF$ that encode structural relations between the variables. The term ``exogenous'' indicates that any causes of such variables lie outside the model; they can be thought of as local `noise variables'.
\begin{definition}[Classical Structural Causal Model]\label{def: classical_structural_cm_def}
    A \emph{(classical) structural causal model (CSM)} $\M$ is a triple $\M = \langle \bU, \bV, \bF\rangle$, where $\bV=\lbrace V_1,\dots,V_n\rbrace$ is a set of \red endogenous \blk  variables, $\bU = \{U_1,\cdots,U_n\}$ is a set of \red exogenous \blk variables and $\bF=\lbrace f_1,\dots,f_n\rbrace$ is a set of functions \mf{such that $v_i = f_i(pa_i,u_i)$ for some $Pa_i \subseteq \mathbf{V}$}.
\end{definition}

Every structural causal model $\M$ is associated with a directed graph $G(\M)$, which represents the causal structure of the model as specified by the relations $G(M) \ni (V_j \stackrel{f_i}{\ra} V_i) \Leftrightarrow V_j \in \Pa_i$. Here, we will restrict CSMs to those defining directed acyclic graphs.
For example, the causal model of Fig.~\ref{DAG_basic} can be extended to a CSM with causal relations as depicted in Fig.~\ref{classical_structural}.
\begin{figure}
	\centering
	\begin{tikzpicture}[node distance={24mm}, thick, main/.style = {draw, circle}] 
\node[main] (1) {$U_1$};
\node[main] (2) [right=0.5cm of 1]{$V_1$}; 
\node[main] (3) [above right of=2] {$V_2$}; 
\node[main] (4) [left=0.5cm of 3] {$U_2$};
\node[main] (5) [below right of=2] {$V_3$}; 
\node[main] (6) [above right of=5] {$V_4$}; 
\node[main] (7) [above right of=6] {$V_5$}; 
\node[main] (8) [below right of=6] {$V_6$}; 
\node[main] (9) [left=0.5cm of 5] {$U_3$};
\node[main] (10) [right=0.5cm of 6] {$U_4$};
\node[main] (11) [right=0.5cm of 7] {$U_5$};
\node[main] (12) [right=0.5cm of 8] {$U_6$};
\draw[->] (1) -- (2); 
\draw[->] (2) -- (3);  
\draw[->] (2) -- (5); 
\draw[->] (4) -- (3);
\draw[->] (9) -- (5);
\draw[->] (3) -- (6); 
\draw[->] (5) -- (6); 
\draw[->] (10) -- (6);
\draw[->] (8) -- (6);
\draw[->] (6) -- (7);
\draw[->] (11) -- (7);
\draw[->] (12) -- (8);
\end{tikzpicture}  
		\caption{A classical structural causal model (CSM) with endogenous nodes $\bV = \{V_1,\cdots,V_6\}$ and exogenous nodes $\bU = \{U_1,\cdots,U_6\}$.} 
		\label{classical_structural}
\end{figure}
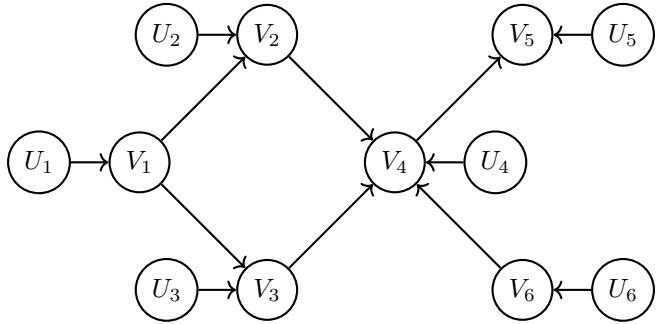

In analogy with the do-interventions for causal Bayesian networks in Sec.~\ref{sec: causal Bayesian networks}, we define do-interventions in a CSM $\M = \langle \bU,\bV,\bF\rangle$. Let $\bX\subset \bV$ with corresponding exogenous variables $\bU(\bX) \subset \bU$ and functions $\bF(\bX) \subset \bF$, and let $\bV_\bx = \bV\backslash \bX$, $\bU_\bx = \bU\backslash \bU(\bX)$ and $\bF_\bx =\bF\backslash \bF(\bX)$. Then the do-intervention $\mathrm{do}(\bX=\bx)$ defines a submodel $\M_\bx = \langle \bU_\bx,(\bV_\bx,\bX=\bx),\bF_\bx\rangle$. In terms of the causal graph $G(\M)$, the action $\mathrm{do}(\bX=\bx)$ removes all incoming arrows to the nodes $X_i$, thus generating a new graph $G(\M_\bx)$.

The submodel $\M_\bx$ represents a \emph{minimal change} to the original model $\M$ such that $\bX=\bx$ is true while keeping the values of the exogenous variables fixed -- which are thought of as ``background conditions''. In turn, we can use $\M_\bx$ to analyze counterfactual statements with antecedent $\bX=\bx$.

\begin{definition}[Counterfactual]\label{def: counterfactual}
    Let $\M = \langle \bU, \bV, \bF\rangle$ be a structural causal model, and let $\bX,\bY \subseteq \bV$. The counterfactual statement ``$\bY$ would have been $\by$, had $\bX$ been $\bx$, in a situation specified by the background variables $\bU=\bu$" is denoted by $\bY_\bx(\bu) = \by$, where $\bY_\bx(\bu)$ is the potential response of $\bY$ to the action $\mathrm{do}(\bX=\bx)$, that is, the solution for $\bY$ of the modified set of equations $\bF_\bx$ in the submodel $\M_\bx$. $\bX=\bx$ is called the \emph{antecedent} and $\bY=\by$ is the \emph{consequent} of the counterfactual.
\end{definition}

Note that given any complete specification $\bU=\bu$ of the exogenous variables, every counterfactual statement of the form above has a truth value.\footnote{\mf{Here, by ``truth values'' we mean \emph{Boolean} truth values `true' and `false'. We do not consider more general truth values such as `indefinite' which may arise e.g. in intuitionistic logic.}} Denoting a ``causal world'' by the pair $\langle \M, \bu\rangle$, we can say that a counterfactual has a truth value in every causal world where it can be defined. This is the case even when the model $\M$ with $\bU=\bu$ determines $\bX$ to have a value different from that specified in the antecedent, because the counterfactual is evaluated relative to the modified submodel $\M_\bx$.

\begin{definition}[Probabilistic structural causal model]\label{def: prob_cscm}
    A \emph{probabilistic structural causal model (PSM)} is defined by a pair $\langle \M, P(\bu)\rangle$, where $\M = \langle \bU,\bV,\bF \rangle$ is a structural causal model (see Def ~\ref{def: classical_structural_cm_def}) and $P(\bu)$ is a probability distribution defined over the exogenous variables $\bU$ of $\M$.
\end{definition}

Since every endogenous variable $V_i \in \bV$ is a function of $U_i$ and its parent nodes, $f_i: U_i \times \Pa_i \ra V_i$, the distribution $P(\bu)$ in a PSM $\langle \M,P(\bu)\rangle$ defines a probability distribution over every subset $\bY \subseteq \bV$ by
\begin{equation}
    P(\by) := P(\bY=\by) = \sum_{\bu|\bY(\bu)=\by} P(\bu) \;.
\end{equation}
In particular, the \emph{probability of the counterfactual} ``$\bY$ would have been $\by$, had $\bX$ been $\bx$'' can be computed using the submodel $\M_\bx$ as
\begin{equation}\label{eq: counterfactual - prediction}
    P(\bY_\bx=\by) = \sum_{\bu|\bY_\bx(\bu)=\by} P(\bu)\; .
\end{equation}

More generally, the probability of a counterfactual query might be conditioned on prior observations `$\be$'. In this case, we first update the probability distribution $P(\bu)$ in the PSM to obtain a modified probability distribution $P(\bu\mid \be)$ conditioned on observed data $\be$ and then use this updated probability distribution to evaluate the probability for the counterfactual as in Eq.~(\ref{eq: counterfactual - prediction}). Combining the above steps, one arrives at the following theorem, proved in Ref.~\cite{Pearl2000}:
	
\begin{theorem}[Pearl \cite{Pearl2000}]\label{thm: evaluation of classical counterfactuals}
	Given a probabilistic structural causal model (PSM) $\langle \M,P(\bu)\rangle$ (see Def ~\ref{def: prob_cscm}), and subsets $\bX,\bY,\bE \subset \bV$, the probability for the counterfactual ``$\bY$ would have been $\by$, had $\bX$ been $\bx$'', given the observation of $\bE=\be$, is denoted by $P(\bY_\bx|\be)$ and can be evaluated systematically by a three-step procedure:
		\begin{itemize}
			\item Step 1: \textbf{Abduction}: using the observed data $\bE=\be$, use Bayesian inference to update the probability distribution $P(\bu)$ in the PSM $\langle \M,P(\bu)\rangle$ to obtain $P(\bu|\be)$.
			\item Step 2: \textbf{Action}: perform a do-intervention $\mathrm{do}(\bX=\bx)$, by which the values of $\bX \subset \bV$ are specified independently of their parent nodes. The resultant model is denoted as $\M_\bx$.
			\item Step 3: \textbf{Prediction}: in the modified model $\langle \M_\bx, P(\bu|\be)\rangle$, compute the probability of $\bY$ via Eq.~(\ref{eq: counterfactual - prediction}).
	\end{itemize} 
\end{theorem}

As an example, consider the situation where $\bX=\bx$ and $\bY=\by$ are observed, that is, $\bE=(\bX,\bY)$.\footnote{Note that $\bX,\bE$ and $\bY,\bE$ in Thm.~\ref{thm: evaluation of classical counterfactuals} are not necessarily disjoint.} We evaluate the probability of the counterfactual ``$\bY$ would have been $\by'$, had $\bX$ been $\bx'$" as:
\begin{align}\label{eq: prediction step in CSM}
    P(\bY_{\bx'}=\by'|\bX=\bx, \bY=\by)
    &= \sum_{\bu|\bY_{\bx'}(\bu)=
    \by'} P(\bu|\bX=\bx,\bY=\by) = \sum_{\bu|\bY_{\bx'}(\bu)=\by'} \frac{P(\bX=\bx,\bY=\by|\bu)P(\bu)}{P(\bX=\bx,\bY=\by)}\; ,
\end{align}
where we used Eq.~(\ref{eq: counterfactual - prediction}) in the first and Bayes' theorem in the second step.\footnote{Note that a probabilistic structural causal model implies the existence of a joint probability distribution over all variables. In this case, an alternative expression for the probability of the counterfactual in Eq.~(\ref{eq: prediction step in CSM}) reads
\begin{align}\label{eq: prediction based on joint distribution}
    P(\bY_{\bx'}=\by'|\bX=\bx, \bY=\by)
    = \dfrac{P(\bY_{\bx'}=\by', \bX=\bx, \bY=\by)}{P(\bX=\bx, \bY=
    \by)}\; .
\end{align}
In the quantum case such a distribution does not generally exist \cite{KochenSpecker1967}.} 

In temporal metaphors, step 1 explains the past (the exogenous variables $\bU$) in light of the current evidence $\be$; step 2 minimally bends the course of history to comply with the hypothetical antecedent and step 3 predicts the future based on our new understanding of the past and our newly established condition.

\section{Quantum violations of classical causality}\label{sec: quantum violations}

Classical causal models face notorious difficulties in explaining quantum correlations. Firstly, Bell's theorem \cite{Bell1964,Bell1975,WisemanCavalcanti2017} can be interpreted in terms of classical causal models, as proving that such models cannot reproduce all quantum correlations (in particular, those that violate a Bell inequality) while maintaining relativistic causal structure and the assumption of ``free choice''. The latter is the assumption that experimentally controllable parameters like measurement settings can always be chosen via ``free variables'', which can be understood as variables that have no \emph{relevant causes} in a causal model for the experiment. That is, they share no common causes with, nor are caused by, any other variables in the model. Thus, ``free variables'' can be modeled as exogenous variables.

For concreteness, consider the standard Bell scenario with a causal structure represented in the DAG in Fig.~\ref{Bell}, where variables $A$ and $B$ denote the outcomes of experiments performed by two agents, Alice and Bob. Variables $X$ and $Y$ denote their choices of experiment, which are assumed to be ``free variables'' and thus have no incoming arrows. Since Alice and Bob perform measurements in space-like separated regions, no relativistic causal connection is allowed between $X$ and $B$ nor between $Y$ and $A$. In this scenario, Reichenbach's principle of common cause \cite{Reichenbach1991,cavalcanti2014modifications} -- which is a consequence of the classical causal Markov condition -- implies the existence of common causes underlying any correlations between the two sides of the experiment. $\Lambda$ denotes a complete specification of any such common causes. As we are assuming a relativistic causal structure, those must be in the common past light cone of Alice's and Bob's experiments.
	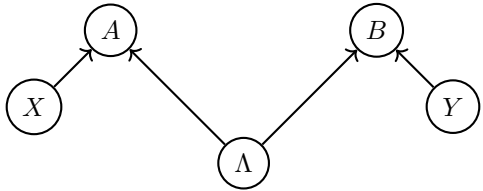
\begin{figure}
		\centering
		\begin{tikzpicture}[node distance={25mm}, thick, main/.style = {draw, circle}] 
                \node[main] (Lambda) {$\Lambda$}; 	
                \node[main] (A) [ above left of=Lambda] {$A$}; 
                \node[main] (B) [ above right of=Lambda] {$B$}; 
                \node[main] [ below left=0.5cm and 0.5cm of A] (X) {$X$};
			\node[main] (Y) [ below right=0.5cm and 0.5cm of B] {$Y$};
			\draw[->] (Lambda) -- (A); 
			\draw[->] (Lambda) -- (B);
			\draw[->] (X) -- (A); 
			\draw[->] (Y) -- (B);
		\end{tikzpicture} 
		\caption{A Directed Acyclic Graph (DAG) depicting the standard Bell scenario.} \label{Bell}
\end{figure}
	
Marginalizing over the common cause variable $\Lambda$, the classical causal Markov condition applied to the DAG in Fig.~\ref{Bell} implies the factorization:
\begin{equation}\label{eq: factorisability}
		P(AB|XY) = \sum_{\Lambda} P(\Lambda) P(A|X\Lambda)P(B|Y\Lambda)\; .
\end{equation}
A model satisfying Eq.~(\ref{eq: factorisability}) is also called a \emph{local hidden variable model}. Importantly, local hidden variable models satisfy the Bell inequalities \cite{Bell1964,Bell1975}, which have been experimentally violated by quantum correlations \cite{CHSH1969,AspectEtAl1981,ZeilingerEtAl2015,ShalmEtAl2015}.\footnote{The 2022 Nobel Prize in Physics was awarded in part for the demonstration of Bell inequality violations.} It follows that no classical causal model can explain quantum correlations under the above assumptions.

More recently, Wood and Spekkens \cite{wood2015lesson} showed that certain Bell inequality violations cannot be reproduced by any classical causal model that satisfies the assumption of ``no fine-tuning''. This is the requirement that any conditional independence between variables in the model be explained as arising from the structure of the causal graph, rather than from fine-tuned model parameters. This assumption is essential for causal discovery -- without it, it is generally not possible to experimentally determine which of a number of candidate graphs faithfully represents a given situation. This result was later generalized to arbitrary Bell and Kochen-Specker inequality violations in Refs.~\cite{cavalcanti2018classical, pearl2019classical}.

These results motivate the search for a generalization of classical causal models that accommodates quantum correlations and allows for causal discovery, while maintaining faithfulness to relativistic causal structure. Ref.~\cite{cavalcanti2014modifications} considers modifications of Reichenbach's principle of common cause \cite{Reichenbach1991}---which is implied by the causal Markov condition in the special case of the common cause scenario in Fig.~\ref{Bell}, as assumed in Bell's theorem \cite{Bell1964,Bell1975}. The authors of Ref.~\cite{cavalcanti2014modifications} argue that one could maintain the \emph{principle of common cause}---the requirement that correlations between two causally disconnected events should be explained via common causes---by relaxing the condition that a full specification of those common causes factorizes the probabilities for the events in question, as by Eq.~(\ref{eq: factorisability}). Using the Leifer-Spekkens formalism for \emph{quantum conditional states}, they instead propose that Eq.~(\ref{eq: factorisability}) should be replaced by the requirement that the \emph{channels} between the common cause and Alice and Bob's labs factorize---or more precisely, the Choi-Jamio{\l}kowski operators corresponding to those channels. This is essentially the type of resolution of Bell's theorem that is provided by \emph{quantum causal models}, to which we now turn. After introducing structural quantum causal models in Sec.~\ref{sec: QSM} and quantum counterfactuals queries in Sec.~\ref{sec: counterfactuals in QSM}, in Sec.~\ref{sec: Bell scenario} we will revisit the Bell scenario from the perspective of counterfactuals in quantum causal models.

\section{Quantum causal models}\label{sec: quantum causal models}

In recent years a growing number of papers have addressed the problem of generalizing the classical causal models formalism to accommodate quantum correlations, in a way that is compatible with relativistic causality and faithfulness. This has led to the development of various frameworks for quantum causal models. The more developed of those are the frameworks by Costa and Shrapnel \cite{costa2016quantum} and Barrett, Lorenz and Oreshkov \cite{LorenzOreshkovBarrett2019}. In this work, we use a combination of the notation and features of both of these formalisms.

\textbf{Quantum nodes and quantum interventions.} Recall that in a classical causal model, a node 
represents a locus for potential interventions. In order to generalize this to the quantum case, 
we start by introducing a \emph{quantum node} $A$, which is associated with two Hilbert spaces $\cH_{A^{\mathrm{in}}}$ and $\cH_{A^{\mathrm{out}}}$, corresponding to the incoming system and the outgoing system, respectively. An intervention at a quantum node $A$ is represented by a \emph{quantum instrument} $\mathcal{I}^z_A$ (see Fig.~\ref{Quantum node}).  This is a set of trace-non-increasing completely positive (CP) maps from the space of linear operators on $\cH_{A^{\mathrm{in}}}$ to the space of linear operators on $\cH_{A^{\mathrm{out}}}$,
\begin{equation}
    \mathcal{I}_A^z = \lbrace\mathcal{M}_A^{a|z}:\mathcal{L}(\Ain)\rightarrow \mathcal{L}(\Aout)\rbrace_a\; ,
\end{equation}
such that $\mathcal{M}_A = \sum_a\mathcal{M}_A^{a|z}$ is a completely positive, trace-preserving (CPTP) map---i.e. a \emph{quantum channel}.\footnote{We sometimes write $\mathcal{M}_A^{|z}$ for this CPTP map to indicate that it is associated with the instrument $\mathcal{I}_A^z$. Note however that a given CPTP map will in general be associated with many different instruments.} Here, $z$ is a label for the (choice of) instrument, and
$a$ labels the classical outcome of the instrument, which occurs with probability $P_z(a) = \mathrm{Tr}[\mathcal{M}_A^{a|z}(\rho_{\Ai})]$ for an input state $\rho_{\Ai} \in \mc{L}(\Ain)$; consequently, the state on the output system conditioned on the outcome of the intervention is given by $\mathcal{M}_A^{a|z}(\rho_{\Ai})/P_z(a)$. For simplicity, we consider finite-dimensional systems only.

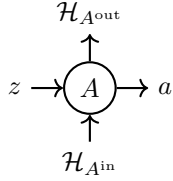
\begin{figure}[h]
		\centering
		\begin{tikzpicture}[node distance={10mm}, thick, main/.style = {draw, circle}] 
			\node[main] (1) {$A$}; 
			\node[] (2) [ left of=1] {$z$}; 
			\node[] (3) [ right of=1] {$a$}; 
			\node[] (4) [ above of=1] {$\mathcal{H}_{\Ao}$}; 
			\node[] (5) [ below of=1] {$\mathcal{H}_{\Ai}$}; 
 			\draw[->] (2) -- (1); 
 			\draw[->] (1) -- (3);
 			\draw[->] (1) -- (4); 
			\draw[->] (5) -- (1);
		\end{tikzpicture} 
		\caption{A quantum node $A$ is associated with an incoming ($\mc{H}_{\Ai}$) and outgoing Hilbert space ($\mc{H}_{\Ao})$. It can be intervened on via a quantum instrument $\mathcal{I}_A^z=\{\cM_A^{a|z}\}_a$, resulting in an outcome `$a$' corresponding to a completely-positive (CP) map $\cM_A^{a|z}$.} \label{Quantum node}
\end{figure}

Using the Choi-Jamio\l kowski (CJ) isomorphism,\footnote{Here, we follow the notation in Ref.~\cite{costa2016quantum}. This differs from the one used in Refs.~\cite{AllenEtAl2017,HoffreumonOreshkov2021,LorenzOreshkovBarrett2019}, which applies a basis-independent version of the Choi-Jamio\l kowski isomorphism, by identifying the Hilbert space associated with outgoing systems with its dual (see also Ref.~\cite{FrembsCavalcanti2022}).} 
we represent a quantum instrument $\mathcal{I}_A^z=\{\cM_A^{a|z}\}_a$ in terms of a positive operator-valued measure $a \mapsto \tau_A^{a|z}$. More precisely, every completely positive map $\cM_A^{a|z}$ is represented by a positive semi-definite operator $\tau_A^{a|z} \in \mc{L}(\cH_{\Ao} \otimes \cH_{\Ai})$ given by
\begin{equation}\label{eq: CIJ for interventions}
	\tau^{a|z}_{A} = \sum_{i,j} \cM_A^{a|z}(|i\rangle \langle j|)^T_{\Ao}\otimes |j\rangle \langle i|_{\Ai}\; .
\end{equation}
In a slight abuse of notation, we will write $\mc{I}^z_A = \{\mc{M}^{a|z}_A\}_a \stackrel{CJ}{\leftrightarrow} \{\tau^{a|z}_A\}_a$ also for the representation of an instrument in terms of positive operators under the Choi-Jamio\l kowski isomorphism. Note that the fact that $\cM_A = \sum_a \cM_A^{a|z}$ is trace-preserving imposes the following trace condition on $\tau^{|z}_A=\sum_a \tau^{a|z}_A$ (cf. Ref.~\cite{Jamiolkowski1972}),
\begin{equation}\label{eq: trace-preserving Choi matrices}
    \tr_{\Ao}[\tau^{|z}_{A}] = \mathbb{I}_{\Ai}\; .
\end{equation}

\textbf{Quantum process operators.} In a quantum causal model we will distinguish between two types of quantum operations: quantum interventions, which are local to a quantum node, and a quantum process operator, which acts between quantum nodes and contains information about the causal (influence) relations between the nodes in the model.

To motivate the general definition (Def.~\ref{def: process operator} below), we first consider the simplest case: for a single quantum node $A$, a quantum process operator is any operator $\sigma_A \in \mc{L}(\cH_{\mathrm{A}^{\mathrm{in}}} \otimes \cH_{\mathrm{A}^{\mathrm{out}}})$ such that the pairing\footnote{With Ref.~\cite{LorenzOreshkovBarrett2019}, we will adopt the shorthand $\tr_A[\cdots] := \tr_{\Ai\Ao}[\cdots]$.}
\begin{equation}\label{eq: gen Born rule bipartite process}
    \tr_A[\sigma_A \tau^{a|z}_A] = \tr_{\Ai\Ao}[\sigma_A \tau^{a|z}_A] =: P_z(a) \in [0,1]\; , 
\end{equation}
defines a probability for every positive semi-definite operator $\tau^{a|z}_A$, and satisfies the normalisation condition
\begin{equation}\label{eq: normalisation condition bipartite process operator}
    \sum_a P_z(a) = \tr_A[\sigma_A \tau^{|z}_A] = 1\; ,
\end{equation}
for every quantum channel (CPTP map) $\tau^{|z}_A$. Consequently, given a process operator $\sigma_A$, we may interpret $P_z(a)$ as the probability to obtain outcome $a$ when performing an instrument $z$.

As a generalisation of the Born rule (on the composite system $\cH_\Ai \otimes \cH_\Ao$), Eq.~(\ref{eq: gen Born rule bipartite process}) in particular implies that $\sigma_A$ is positive, hence, corresponds to a completely positive map $\mathcal{E}: \mathcal{L}(\Aout) \ra \mathcal{L}(\Ain)$. 

More generally, it will be useful to introduce a notation for the positive semi-definite operator $\rho^\mathcal{E}_{\mf{B|A}}$ corresponding to a bipartite channel of the form $\mathcal{E}: \mathcal{L}(\Aout) \ra \mathcal{L}(\Bin)$:
\begin{equation}\label{eq: CIJ for processes}
    \rho^\mathcal{E}_{B|A} = \rho^\mathcal{E}_{\Bi|\Ao} := \sum_{i,j} \mathcal{E}(|i\rangle \langle j|)_{\Bi}\otimes |i\rangle_{\Ao} \langle j|\; .
\end{equation}
Note that $\rho^\mathcal{E}_{B|A}$ is distinguished from the representation of the Choi matrices corresponding to quantum instruments in Eq.~\eqref{eq: CIJ for interventions} by an overall transposition, indicating the different roles played by instruments and processes in the inner product of Eq.~(\ref{eq: gen Born rule bipartite process}). In particular, we have $\sigma_A = \rho^\mathcal{E}_{A|A} = \rho^\mathcal{E}_{\Ai|\Ao}$ for some channel satisfying the normalisation condition in Eq.~(\ref{eq: normalisation condition bipartite process operator}).

Generalizing this idea to finitely many quantum nodes, a quantum process operator is defined as follows.

\begin{definition}[Process operator]\label{def: process operator} A \emph{(quantum) process operator} over quantum nodes $\bA = \{A_1, \cdots,A_n\}$ is a positive semi-definite operator $\sigma_\bA = \sigma_{A_1, \cdots, A_n} \in \mathcal{L}(\bigotimes_{i=1}^n \mathcal{H}_{A_i^{\ins}} \otimes  \mathcal{H}_{A_i^{\out}})_+$, which satisfies the normalisation condition,
\begin{equation}\label{eq: normalisation quantum process operator}
    \mathrm{Tr}_{A_1\cdots A_n}[\sigma_{A_1, \cdots, A_n} (\tau^{|z_1}_{A_1}\otimes\cdots\otimes\tau^{|z_n}_{A_n})]=1\; ,
\end{equation}
for any choice of quantum channels $\tau^{|z_1}_{A_1},\cdots,\tau^{|z_n}_{A_n}$ at nodes $A_1,\cdots,A_n$.\footnote{Every process operator satisfies a trace condition analogous to Eq.~(\ref{eq: trace-preserving Choi matrices}): $\mathrm{Tr}_{A_1^{\ins}\cdots A_n^{\ins}}[\sigma_{A_1, \cdots, A_n}] = \mathbb{I}_{A_1^{\out}}\otimes\cdots\otimes\mathbb{I}_{A_n^{\out}}$, hence, $\sigma_{A_1\cdots A_n}$ defines a CPTP map $\mc{L}(\cH_{A^{\out}_1}\otimes \cdots \otimes \cH_{A^{\out}_n}) \ra \mc{L}(\cH_{A^{\ins}_1}\otimes \cdots \otimes \cH_{A^{\ins}_n})$. Yet, the converse is generally not true.}
\end{definition}

Comparing with Eq.~(\ref{eq: gen Born rule bipartite process}), we define the probability of obtaining outcomes $\ba = \{a_1,\cdots,a_n\}$ when performing interventions $\{\{\tau_{A_1}^{a_{1}|z_1}\}_{a_{1}},\cdots,\{\tau_{A_n}^{a_{n}|z_n}\}_{a_{n}}\}$ at quantum nodes $\bA = \{A_1,\cdots,A_n\}$ by
\begin{equation}\label{eq: generalised Born rule}
    P_\bz(\ba) = \mathrm{Tr}_{A_1\cdots A_n}[\sigma_{A_1, \cdots, A_n}(\tau_{A_1}^{a_{1}|z_1}\otimes\cdots\otimes\tau_{A_n}^{a_{n}|z_n})]\; .
\end{equation}
Eq.~(\ref{eq: generalised Born rule}) defines a generalization of the Born rule (on the composite system $\bigotimes_{i=1}^n \mathcal{H}_{A_i^{\ins}} \otimes  \mathcal{H}_{A_i^{\out}}$) \cite{AraujoEtAl2015,ShrapnelCostaMilburn2018}.

\vspace{0.2cm}

\textbf{Quantum causal models.} With the above ingredients, we obtain quantum generalizations of the causal Markov condition in Def.~\ref{def: cmc} and thereby of classical causal models (causal networks) in Def.~\ref{def: CCM}.

\begin{definition}[Quantum causal Markov condition]\label{def: qmc}
    A quantum process operator $\sigma_\bA = \sigma_{A_1, \cdots, A_n}$ is \emph{Markov} for a given DAG $G$ if and only if there exist \mf{positive} operators $\rho_{A_i|\Pa(A_i)}$ \mf{such that $\tr_{\Pa_i^\out}[\rho_{A_i|\Pa(A_i)}] = \mathbb{I}_{A_i^\ins}$ (corresponding to quantum channels $\mathcal{E}_i: \mathcal{L}(\cH_{\Pa_i^\mathrm{out}}) \ra \mathcal{L}(\cH_{A_i^\mathrm{in}})$)} for each quantum node $A_i$ of $G$ such that\footnote{Here and below, we implicitly assume the individual operators $\rho_{A_i|\Pa(A_i)}$ to be `padded' with identities on all nodes not explicitly involved in $\rho_{A_i|\Pa(A_i)}$ such that the multiplication of operators is well-defined.}
    \begin{equation}
	\sigma_\bA = \prod_{i=1}^{n} \rho_{A_i|\Pa(A_i)}\; ,
    \end{equation}
    and $[\rho_{A_i|\Pa(A_i)},\rho_{A_j|\Pa(A_j)}] = 0$ for all $i,j \in \{1,\cdots,n\}$.
\end{definition} 

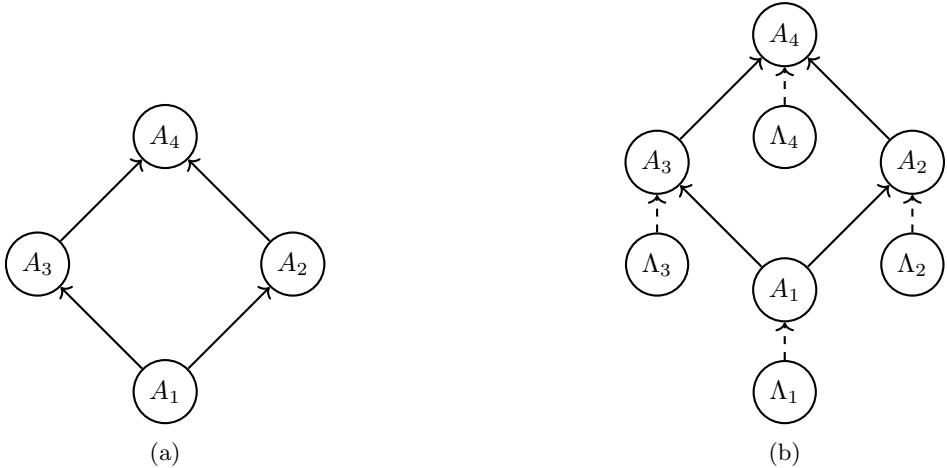
\begin{figure}[h]
    \begin{minipage}[b]{0.49\textwidth}
        \centering
        \begin{tikzpicture}[node distance={24mm}, thick, main/.style = {draw, circle}] 
            \node[main] (1) {$A_1$}; 
            \node[main] (2) [ above left of=1] {$A_3$}; 
            \node[main] (3) [ above right of=1] {$A_2$}; 
            \node[main] (4) [ above right of=2] {$A_4$};
            \draw[->] (1) -- (2); 
            \draw[->] (1) -- (3);
            \draw[->] (2) -- (4); 
            \draw[->] (3) -- (4);
        \end{tikzpicture}
        \subcaption{}
    \end{minipage}
    \begin{minipage}[b]{0.49\textwidth}
	\centering
	\begin{tikzpicture}[node distance={24mm}, thick, main/.style = {draw, circle}] 
		\node[main] (A1) {$A_1$}; 
		\node[main] (A2) [ above right of=1] {$A_2$};
            \node[main] (A3) [ above left of=A1] {$A_3$}; 
            \node[main] (A4) [ above right of=A3] {$A_4$};
			\node[main] (Lambda1) [below=0.5cm of A1] {$\Lambda_1$}; 
			\node[main] (Lambda2) [below=0.5cm of A2] {$\Lambda_2$}; 
			\node[main] (Lambda3) [below=0.5cm of A3] {$\Lambda_3$};
            \node[main] (Lambda4) [below=0.5cm of A4] {$\Lambda_4$};
			
 			\draw[->] (A1) -- (A2); 
 			\draw[->] (A1) -- (A3);
 			\draw[->] (A2) -- (A4); 
 			\draw[->] (A3) -- (A4);
 			\draw[dashed,->] (Lambda1) -- (A1); 
 			\draw[dashed,->] (Lambda2) -- (A2);
 			\draw[dashed,->] (Lambda3) -- (A3); 
 			\draw[dashed,->] (Lambda4) -- (A4);
 	\end{tikzpicture}
	\subcaption{}
    \end{minipage}
    \caption{Example of a quantum causal model, (a) with four endogenous nodes $\bA = \{A_1, A_2, A_3, A_4\}$, (b) including exogenous nodes $\bm{\Lambda} = \{\Lambda_1,\Lambda_2,\Lambda_3,\Lambda_4\}$ in a quantum structural causal model.}
    \label{fig: example quantum causal model w/o endogenous variables}
\end{figure}

\begin{definition}[Quantum causal model] A \emph{quantum causal model} is a pair $\langle G, \sigma_\bA \rangle$, consisting of a DAG $G$, whose vertices represent quantum nodes $\bA = \{A_1,\cdots,A_n\}$, and a quantum process operator that is Markov with respect to $G$, according to Def.~\ref{def: qmc}.
\end{definition}

\subsection{Quantum structural causal models}\label{sec: QSM}
Recall that in the classical case, counterfactuals are evaluated relative to a classical structural causal model (CSM) $\M = \langle \bU,\bV,\bF\rangle$ (see Def.~\ref{def: classical_structural_cm_def}), which associates an exogenous variable $U_i \in \bU$ and a function $\bF \ni f_i: \Pa(A_i) \times U_i \rightarrow A_i$, to every node $V_i \in \bV$. Given a CSM, we thus have full information about the underlying process and any uncertainty arises solely from our lack of knowledge about the values of the variables at exogenous nodes, which is encoded in the probability distribution $P(\bu)$ of the probabilistic structural causal model (PSM) $\langle \M,P(\bu)\rangle$.
 
In order to define a notion of quantum structural causal models, we find it useful to introduce the lack of knowledge on exogenous nodes directly in terms of a special type of quantum instruments,\footnote{Here, our formalism diverges from the one in Ref.~\cite{LorenzOreshkovBarrett2019}, which assigns the lack of knowledge about exogenous degrees of freedom as part of the process operator $\sigma$, and which does not distinguish between different state preparations. This is a change in perspective in so far as we will place our lack of knowledge as a lack of knowledge about \emph{events} at the exogenous nodes, rather than a lack of knowledge about the \emph{process}.}
\begin{equation}\label{eq: exogenous preparation instruments}
	\{\tau^{\lambda}_\Lambda\}_\lambda := \{P(\lambda) (\rho^{\lambda}_{\Lambda^{\out}})^T \otimes \mathbb{I}_{\Lambda^{\ins}}\}_\lambda \; .
\end{equation}
Quantum instruments of this form discard the input to the node $\Lambda$ and with probability $P(\lambda)$ prepare the state $\rho^\lambda$ in the output. In other words, $\{\tau^{\lambda}_\Lambda\}_\lambda$ is a discard-and-prepare instrument. Ignoring the outcome of this instrument, one obtains the channel $\tau^{\rho}_\Lambda = \sum_\lambda \tau^{\lambda}_\Lambda$, corresponding to the preparation of state $\rho = \sum_\lambda P(\lambda) \rho^{\lambda}$ in the output of node $\Lambda$.

Note that the outcome and output of a discard-and-prepare instrument are independent of the input state $\rho_{\Lambda^\ins}$. In order to avoid carrying around arbitrary input states in formulas below (as required for normalization), we will therefore adopt the convention,
\begin{equation}\label{eq: exogenous preparation instruments}
	\{\widetilde{\tau}^{\lambda}_\Lambda\}_\lambda := \left\{P(\lambda) (\rho^{\lambda}_{\Lambda^{\out}})^T \otimes \frac{1}{\mathrm{dim}(\cH_{\Lambda^\ins})}\mathbb{I}_{\Lambda^{\ins}}\right\}_\lambda \; ,
\end{equation}
such that $\tr_{\Lambda^\ins}[\widetilde{\tau}^\lambda_\Lambda] = \tr_{\Lambda^\ins}[\tau^\lambda_\Lambda \rho_{\Lambda^\ins}]$ for any state $\rho_{\Lambda^\ins}$.

\begin{definition}(no-influence condition).\label{def: no-influence conditions}
    Let $\rho_{CD|AB}^{U}$ be the Choi-Jamiolkowski (CJ) representation of the channel corresponding to the unitary transformation $U: \mathcal{H}_A \otimes \mathcal{H}_B \rightarrow \mathcal{H}_C \otimes \mathcal{H}_D$. We say that system $A$ does not influence system $D$ (denoted as $A \nrightarrow D$) if and only if there exists a quantum channel $\cM:\mathcal{L}(\mathcal{H}_B) \rightarrow \mathcal{L}(\mathcal{H}_D)$ with corresponding CJ representation $\rho_{D|B}^{\cM}$ such that $\tr_C[\rho_{CD|AB}^{U}] = \rho_{D|B}^\cM\otimes \mathbb{I}_{A}$.\footnote{We remark that the labels $A,B,C,D$ refer to arbitrary systems, not necessarily nodes in a quantum causal model. Within a quantum causal model, two of those labels, say $A$ and $C$, may refer to output and input Hilbert spaces of the same node.}
\end{definition}

Given these preliminaries, we define a quantum version of the structural causal models in Def.~\ref{def: classical_structural_cm_def}.

\begin{definition}[Quantum structural causal model]\label{def: quantum structural causal model}
    A \emph{quantum structural causal model (QSM)} is a triple\\ $\Q=\langle (\bA, \bL, S), \rho_{\bA S|\bA \bL}^U,\{\tau^{\lambda_i}_{\Lambda_i}\}_{\lambda_i} \rangle$, specified by:
    \begin{itemize}
        \item[(i)] a set of quantum nodes, which are split into
        \begin{itemize}
            \item a set of endogenous nodes $\bA=\{A_1,\cdots,A_n\}$, 
            \item a set of exogenous nodes $\bL = \{\Lambda_1,\cdots,\Lambda_n\}$,
            \item and a sink node $S$;
        \end{itemize}
        \item[(ii)] a unitary $\rho_{\bA S|\bA \bL}^U$ that satisfies the no-influence conditions
        \begin{equation}\label{eq: no-influence conditions I}
            \lbrace \Lambda_j \nrightarrow A_i\rbrace_{j\neq i} 
        \end{equation}
        according to Def.~\ref{def: no-influence conditions}; and 
        \item[(iii)] a set of discard-and-prepare instruments
        $\{\tau^{\lambda_i}_{\Lambda_i}\}_{\lambda_i}$ for every exogenous node $\Lambda_i \in \bL$.
    \end{itemize}
\end{definition}

Note that in general we need to include an additional \emph{sink node} $S$, in order for the process operator $\rho_{\bA S|\bA \bL}^U$ to be unitary. $S$ contains any excess information that is discarded in the process (cf. Ref.~\cite{LorenzOreshkovBarrett2019}).

We emphasize the subtle, but conceptually crucial difference between Def.~4.5 in Ref.~\cite{LorenzOreshkovBarrett2019} and our Def.~\ref{def: quantum structural causal model}. The former specifies the input states on ancillary nodes directly, as part of a `unitary process with inputs', while the latter encodes input states in terms of discard-and-prepare instruments, acting on an arbitrary input state. This will enable us to use classical Bayesian inference on the outcomes of instruments at exogenous nodes in the abduction step of the evaluation of quantum counterfactuals, as we'll see in Sec.~\ref{sec: counterfactuals in QSM} below. In contrast, this is not possible using Def.~4.5 in Ref.~\cite{LorenzOreshkovBarrett2019}, but may require a generalisation of Bayesian inference to the quantum case (see Sec.~\ref{sec: quantum Bayes}).

Following Ref.~\cite{LorenzOreshkovBarrett2019}, we define a notion of \emph{structural compatibility} of a process operator $\sigma_{\bA}$ with a graph $G$.

\begin{definition}\label{def: compatibility}[Compatibility of a quantum process operator with a DAG]
	A quantum process operator $\sigma_{\bA}=\sigma_{A_1\cdots A_n}$ over nodes $\bA=\{A_1,\cdots,A_n\}$ is said to be \emph{structurally compatible with a DAG $G$} if and only if there exists a quantum structural causal model (QSM) $\Q=\langle (\bA, \bL, S), \rho_{\bA S|\bA \bL}^U,\{\tau^{\lambda_i}_{\Lambda_i}\}_{\lambda_i} \rangle$ that recovers $\sigma_{\bA}$ as a marginal,

    \begin{equation}\label{eq: compatibility with process operator}
        \sigma_{\bA} = \tr_{S^{\mathrm{in}}\bL}[\rho_{\bA S|\bA \bL}^U (\widetilde{\tau}^{\rho_1}_{\Lambda_1}\otimes \cdots \otimes \widetilde{\tau}^{\rho_n}_{\Lambda_n})]\; ,
    \end{equation}
    where $\rho_{\bA S|\bA \bL}^U$ satisfies the no-influence relations
    \begin{equation}\label{eq: no-influence conditions II}
        \lbrace A_j \nrightarrow A_i\rbrace_{A_j\notin \Pa(A_i)}\; ,
    \end{equation}
    with $\Pa(A_i)$ defined by $G$.
\end{definition}	

Similar to Thm.~4.10 in Ref.~\cite{LorenzOreshkovBarrett2019}, one shows that a process operator $\sigma_{\bA}$ is structurally compatible with $G$ if and only if it is Markov for $G$.

\begin{theorem}[Equivalence of quantum compatibility and Markovianity]\label{thm: compatibility and Markovianity}
    For a DAG $G$ with nodes $\bA = \{A_1,\cdots,A_n\}$ and a quantum process operator $\sigma_\bA$, the following are equivalent:
    \begin{enumerate}
        \item $\sigma_\bA$ is structurally compatible with $G$.
        \item $\sigma_\bA$ is Markov for $G$.
    \end{enumerate}
\end{theorem}

\begin{proof}
    The difference between our definition of `structural compatibility' in Def.~\ref{def: compatibility} and that of 'compatibility' in Def.~4.8 in Ref.~\cite{LorenzOreshkovBarrett2019} is that the latter applies to a ``unitary process with inputs'' (see Def.~4.5 in Ref.~\cite{LorenzOreshkovBarrett2019}), while Def.~\ref{def: compatibility} applies to a QSM as defined in Def.~\ref{def: quantum structural causal model}. Yet, we show that $\sigma_\bA$ is compatible with $G$ if and only if it is structurally compatible with $G$. The result then follows from the proof of Thm.~4.10 in Ref.~\cite{LorenzOreshkovBarrett2019}.

    First, let $\sigma_\bA$ be compatible with $G$, then by Def.~4.8 in Ref.~\cite{LorenzOreshkovBarrett2019} there exists a unitary process $\rho^U_{\bA S|\bA\bm{\Lambda}}$ that satisfies the no-influence conditions $\lbrace A_j \nrightarrow A_i\rbrace_{A_j\notin \Pa(A_i)}$ and $\lbrace \Lambda_j \nrightarrow A_i\rbrace_{j\neq i}$, and states $\rho_{\Lambda_1} \otimes \dots \otimes \rho_{\Lambda_n}$ such that $\sigma_\bA$ is recovered as a marginal,\begin{equation}\label{marginalization_Barrett}
        \sigma_{\bA} = \mathrm{Tr}_{S^\mathrm{in}\bL^{\out}}\left [\rho^U_{\bA S|\bA \bL}(\rho^T_{\Lambda_1}\otimes\cdots\otimes \rho^T_{\Lambda_n})\right]\;,
    \end{equation}
    where we traced over the inputs of exogenous nodes $\Lambda_i$. Choosing discard-and-prepare measurements $\{\tau^{\lambda_i}_{\Lambda_i}\}_{\lambda_i}$ such that $\tau^{\rho_{\Lambda_i}}_{\Lambda_i} := \sum_{\lambda_i} \tau^{\lambda_i}_{\Lambda_i}$ (cf. Eq.~(\ref{eq: exogenous preparation instruments})), $\langle (\bA,\bm{\Lambda},S),\rho^U_{\bA S|\bA\bm{\Lambda}}, \{\tau^{\lambda_i}_{\Lambda_i}\}_{\lambda_i}\rangle$ defines a QSM (cf. Def.~\ref{def: quantum structural causal model}): in particular, $\rho^U_{\bA S|\bA\bm{\Lambda}}$ satisfies Eq.~(\ref{eq: no-influence conditions I}). Moreover, $\rho^U_{\bA S|\bA\bm{\Lambda}}$ also satisfies Eq.~(\ref{eq: no-influence conditions II}), and Eq.~(\ref{marginalization_Barrett}) implies Eq.~(\ref{eq: compatibility with process operator}. From this it follows that $\sigma_\bA$ is structurally compatible with $G$.
    
    Conversely, if $\sigma_\bA$ is structurally compatible with $G$ it admits a QSM $\langle (\bA,\bm{\Lambda},S),\rho^U_{\bA S|\bA\bm{\Lambda}}, \{\tau^{\lambda_i}_{\Lambda_i}\}_{\lambda_i}\rangle$, from which we extract the unitary process operator $\rho^U_{\bA S|\bA\bm{\Lambda}}$ satisfying the no-influence conditions in Eq.~(\ref{eq: no-influence conditions I}) and Eq.~(\ref{eq: no-influence conditions II}), and which recovers $\sigma_\bA$ as a marginal in Eq.~(\ref{marginalization_Barrett}) for inputs $\rho_{\Lambda_i} = \tr_{\Lambda^{\ins}}[\widetilde{\tau}^{\rho_{\Lambda_i}}_{\Lambda_i}] = \sum_{\lambda_i} \tr_{\Lambda^{\ins}}[\widetilde{\tau}^{\lambda_i}_{\Lambda_i}]$, as a consequence of Eq.~(\ref{eq: compatibility with process operator}). It then follows that $\sigma_\bA$ is compatible with $G$.
\end{proof}

Theorem \ref{thm: compatibility and Markovianity} establishes that for every process operator that is Markov for a graph $G$, there exists a QSM model over $G$ that reproduces that process. Note however that this does not necessarily give us information about which QSM \emph{correctly} describes a given physical process. This requires that the outcomes of instruments at the exogenous nodes correspond to ``stable events'' (cf. Ref.~\cite{DiBiagioRovelli2021}), e.g.~due to decoherence. \red That is, for a QSM to be taken to correctly describe a physical process, the events represented at the exogenous nodes must be effectively classical events, in line with their treatment as fixed background events. \blk The evaluation of counterfactuals will be relative to a QSM, and different QSMs compatible with the same process $\sigma_\bA$ will in general give different answers to the same counterfactual query. This situation is analogous to the classical case. The question of determining \emph{which} (classical or quantum) structural causal model correctly describes a given physical realisation of a process is an important question, but beyond the scope of this work.

Finally, we need the following notion (cf. Eq.~(19) in Ref.~\cite{ShrapnelCostaMilburn2018}). Given a particular set of outcomes $\bm{\lambda}=(\lambda_1,\cdots,\lambda_n)$ at the exogenous instruments, we define a \emph{conditional process operator} as follows,
\begin{equation}\label{eq: conditional process operator}
    \sigma_{\bA}^{\bm{\lambda}} = \frac{\tr_{S^{\mathrm{in}}\bL}\left[\rho_{\bA S|\bA \bL}^U 
    (\widetilde{\tau}^{\lambda_1}_{\Lambda_1}\otimes \cdots \otimes \widetilde{\tau}^{\lambda_n}_{\Lambda_n})\right]}{P(\lambda_1,\cdots,\lambda_n)}\; .
\end{equation}
This allows us to calculate the conditional probability `$P_{\bz}(\ba|\bm{\lambda})$' to obtain a set of outcomes $\ba = (a_1,\cdots,a_n)$ for a set of instruments $\bz = (z_1,\cdots,z_n)$ at endogenous nodes, given a set of outcomes $\bm{\lambda}$ for the exogenous instruments:
\begin{equation}\label{eq: marginal likelihood}
    P_{\bz}(\ba|\bm{\lambda}) = \tr_{\bA}[\sigma_{\bA}^{\bm{\lambda}}\tau^{\ba|\bz}_{\bA}] \quad \quad \quad \mathrm{with} \quad \quad  \quad \tau^{\ba|\bz}_{\bA} = \tau_{A_1}^{a_1|z_1} \otimes\cdots\otimes \tau_{A_n}^{a_n|z_n}\; .
\end{equation}
	
Assuming that the a QSM correctly describes a given physical scenario, and in particular that the events associated with $\bm{\lambda}$ can be thought of as well-decohered, stable events, we can think of Eq.~\eqref{eq: conditional process operator} as representing the \emph{actual} process realised in a given run of the experiment, where our (prior) ignorance about which process is actually realised is encoded in the subjective probabilities $P(\lambda_1,\cdots,\lambda_n)$.

\section{Counterfactuals in Quantum Causal Models}\label{sec: counterfactuals in QSM}

Classically, a counterfactual query has the form ``Given evidence $\be$, would $\bY$ have been $\by$ had $\bZ$ been $\bz$?''. In Pearl's formalism, the corresponding counterfactual statement can be assigned a truth value given a full specification $\bU=\bu$ of the background conditions in a structural causal model. In that formalism, probabilities only arise out of our lack of knowledge about exogenous variables, and one can define the probability for the counterfactual to be true as the probability that $\bu$ lies in the range of values where the counterfactual is evaluated as true. In contrast, in quantum causal models, a counterfactual statement will in general not have a truth value! This is the case even if we are given maximal information about the process (represented as a unitary process) and maximal information about the events at the exogenous nodes (represented as a full specification of the exogenous variables `$\bm{\Lambda} = \bm{\lambda}$' in a quantum structural causal model\footnote{Here we are assuming that maximal information about an event corresponding to the preparation of a quantum state is given by a (pure) quantum state. This of course assumes that quantum mechanics is ``complete'' in the sense that there are no hidden variables that would further specify the outcomes of instruments. While this is admittedly an important assumption, it is the natural assumption to make in the context of quantum causal models---which aim to maintain compatibility with relativistic causality~\cite{Cavalcanti2021}.}).

In order to avoid the implicit assumption of `counterfactual definiteness' inherent to the notion of a \emph{probability of a counterfactual} as in the classical case (see Def.~\ref{def: counterfactual}), we seek a notion of \emph{counterfactual probability} in the quantum case. 

\red
\begin{definition}[Counterfactual probability]\label{def: quantum counterfactual prob}
    Let $\Q=\langle (\bA, \bL, S), \rho_{\bA S|\bA \bL}^U,\{\tau^{\lambda_i}_{\Lambda_i}\}_{\lambda_i} \rangle$ be a quantum structural causal model. Then the \emph{counterfactual probability} that outcomes $\bc'$ would have obtained for a subset of nodes \textbf{C}, had instruments $\textbf{z}'=(z_1',\cdots,z_n')$ been implemented and outcomes $\bb'$ obtained at a set of nodes \textbf{B} (disjoint from \textbf{C}), in the situation specified by the background variables $\bm{\Lambda} = \bm{\lambda}$, is denoted by $P^{\bm{\lambda}}_{\textbf{z'}}(\bc'|\bb')$ and given by 
    \begin{equation}\label{eq: prediction cf prob}
    P^{\bm{\lambda}}_{\bz'}(\bc'|\bb')
    = \frac{P^{\bm{\lambda}}_{\bz'}(\bc',\bb')}{P^{\bm{\lambda}}_{\bz'}(\bb')} 
    = \frac{\tr_{\bA} \left[
        \sigma^{\lambda}_{\bA} 
        (\tau^{\bb'|\bz'_\bB}_{\bB} \otimes 
        \tau^{\bc'|\bz'_\bC}_{\bC} \otimes 
        \tau^{|\bz'}_{\bA\setminus\bB\cup\bC})
    \right]}
        {\tr_{\bA}
    \left[
        \sigma^{\lambda}_{\bA} 
        (\tau^{\bb'|\bz'_\bB}_{\bB} \otimes 
        \tau^{|\bz'}_{\bA\setminus\bB})
    \right]}\; ,
\end{equation}
where $\tau^{\bb'|\bz'_\bB}_{\bB}=\bigotimes_{B_j\in \bB}\,\tau^{b'_j|z'_j}_{B_j}$, $\tau^{\bc'|\bz'_\bC}_{\bC}=\bigotimes_{C_k\in \bC}\,\tau^{c'_k|z'_k}_{C_k}$,  $\tau^{|\bz'}_{\bA\setminus\bB\cup\bC}=\bigotimes_{A_i\notin \bB\cup\bC} \tau^{|z'_i}_{A_i}$ and $\tau^{|\bz'}_{\bA\setminus\bB}=\bigotimes_{A_i\notin \bB} \tau^{|z'_i}_{A_i}$. For $P_{\bz'}^{\lambda}(\bb') = 0$, we set $P_{\bz'}^{\lambda}(\bc'|\bb') = *$ for counterfactuals with impossible antecedent (`counterpossibles').
\end{definition}

More generally, we want to calculate the \emph{expected value} of the counterfactual probability given some evidence, for which we define a standard quantum counterfactual query as follows.

\begin{definition}[Standard quantum counterfactual query]\label{def: quantum counterfactual query}
    Let $\Q=\langle (\bA, \bL, S), \rho_{\bA S|\bA \bL}^U,\{\tau^{\lambda_i}_{\Lambda_i}\}_{\lambda_i} \rangle$ be a quantum structural causal model. Then a \emph{standard quantum counterfactual query}, denoted by $P_{\bb'|\bz'}^{\ba|\bz}(\bc')$, is the expected probability that outcomes $\bc'$ would have obtained for a subset of nodes \textbf{C}, had instruments $\textbf{z}'=(z_1',\cdots,z_n')$ been implemented and outcomes $\bb'$ obtained at a set of nodes \textbf{B} (disjoint from \textbf{C}), given the evidence that a set of instruments $\textbf{z}=(z_1,\cdots, z_n)$ has been implemented and outcomes $\ba = (a_1,\cdots, a_n)$ obtained.
\end{definition}
\blk

Note that to obtain an unambiguous answer, one needs to specify all the instruments in all the nodes, both actual and counterfactual. Def.~\ref{def: quantum counterfactual query}
may not look general enough to accommodate all types of counterfactuals one can envisage, but we will discuss later how the answer to seemingly different types of counterfactual queries can be obtained from the answer to a standard query after suitable interpretation. At times there will be ambiguity in how to interpret some counterfactual queries, and the task of interpretation will be to reduce any counterfactual query to the appropriate standard query---we will return to this later. We now proceed to show how we can answer a quantum counterfactual query.

\subsection{Evaluation of counterfactuals}\label{sec:cf_q_evaluation}
The evaluation of a standard counterfactual query within a quantum structural causal model proceeds through a \textit{three-step process of abduction, action and prediction}, in analogy with the classical case.

{\bf Abduction.} We infer what the past must have been, given information we have at present, that is, we want to update our information about the instrument outcomes $\lambda_i$ at the exogenous nodes $\Lambda_i$, given that outcomes $a_i$ have been observed upon performing instruments $z_i$ at nodes $A_i$.\footnote{In the language of Ref.~\cite{DiBiagioRovelli2021}, we treat the outcomes $\lambda_i$ at exogenous nodes $\Lambda_i$ as ``stable facts''. In taking this stance we set aside the question of when an instrument outcome can be said to be a stable fact, i.e.~we set aside the measurement problem, which applies to the quantum causal model framework in the same way as to standard quantum theory~\cite{Cavalcanti2016, Cavalcanti2021}.} 
Since we are talking about jointly measured variables, we can perform Bayesian update to calculate the conditional probability\footnote{Here, we assume that $P_\bz(\ba) > 0$ since $\ba|\bz$ is an actually observed event.} 
\begin{equation}\label{eq: Bayesian update in QSM}
    P_\bz(\bm{\lambda}|\ba)= \frac{P_\bz(\ba|\bm{\lambda})P(\bm{\lambda})}{P_\bz(\ba)} 
    = \frac{\tr_{\bA} \left[ \sigma_{\bA}^{\bm{\lambda}} \tau^{\ba|\bz}_{\bA} \right] P(\lambda_1,\cdots,\lambda_n)}
            {\tr_{\bA} \left[ \sigma_{\bA} \tau^{\ba|\bz}_{\bA} \right]}\; .
\end{equation}

{\bf Action.} Next, we modify the instruments at endogenous nodes to $\{\tau^{a'_i|z'_i}_{A_i}\}_{a'_i}$, as required by the antecedent of the counterfactual query. We highlight an important distinction from the classical case: unlike in Pearl's formalism, we do not need to modify the process itself, since an `arrow-breaking' intervention at a node $A$ can always be emulated via some appropriate discard-and-prepare instrument, for example, by the instrument
\begin{equation}\label{eq: quantum do-intervention}
    \tau^{\mathrm{do}(\rho)}_A := (\rho_{A^{\mathrm{out}}})^T \otimes \mathbb{I}_{A^{\mathrm{in}}}\; .
\end{equation}
Deciding what instruments are appropriate for a given counterfactual query not in standard form is part of the interpretational task we will return to in Sec.~\ref{sec:analysis} below. For a standard quantum counterfactual query, this is unambiguous since the counterfactual instruments are defined as part of the query (see Def.~\ref{def: quantum counterfactual query}).

{\bf Prediction.} Finally, \red we calculate the expected value of the counterfactual probability
\begin{equation}\label{eq: prediction step in QSM}
    P_{\bb'|\bz'}^{\ba|\bz}(\bc') 
    = \sum_{\bm{\lambda} \in \bm{\Lambda}} P_\bz(\bm{\lambda}|\ba) P_{\bz'}^{\bm{\lambda}}(\bc'|\bb')\; .
\end{equation}
 Whenever the counterfactual has an impossible antecedent for some values of the background variables with nonzero probability, that is, whenever $P_{\bz'}^{\bm{\lambda}}(\bc'|\bb') = *$ for some $\bm{\lambda} \in \bm{\Lambda}$ with $P_\bz(\bm{\lambda}|\ba) \neq 0$, we set  $P_{\bb'|\bz'}^{\ba|\bz}(\bc') = *$. \blk 

If a counterfactual query can be interpreted as a standard quantum counterfactual query, then it will have an unambiguous answer as above. In Sec.~\ref{sec:analysis}, we will discuss the task of interpreting a general quantum counterfactual query that is not already in standard form. Before doing so, we proceed by proving that the present formalism extends Pearl's classical formalism.

\section{From classical to quantum structural causal models}\label{sec: quantum extension}

Having defined a notion of quantum structural causal models (QSM) in Def.~\ref{def: quantum structural causal model}, it is an important question to ask in what sense this definition extends that of a probabilistic structural causal model (PSM) in Def.~\ref{def: prob_cscm} and, in particular, that of a classical structural causal model (CSM) in Def.~\ref{def: classical_structural_cm_def}. In this section, we show that QSMs indeed provide a generalization of PSMs---by extending an arbitrary PSM $\langle \M,P(\bu)\rangle$ to a QSM $\Q$. In order to do so, we need to take care of two crucial physical differences between Def.~\ref{def: classical_structural_cm_def} and Def.~\ref{def: quantum structural causal model}.

First, note that the structural relations $\bF$ in a CSM $\M = \langle \bU,\bV,\bF\rangle$ are generally not reversible, while unitary evolution in QSMs postulates an underlying reversible process. We therefore need to lift a generic CSM to a reversible CSM, whose structural relations are given in terms of bijective functions, yet whose independence conditions coincide with those of the original CSM.
Second, while classical information (in a CSM) can be copied, quantum information famously cannot. We therefore need to find a mechanism to encode classical copy operations into a QSM. This will require us to introduce auxiliary systems, which also need to preserve the no-influence conditions required between exogenous variables in Def.~\ref{def: quantum structural causal model}, (ii).

The next theorem asserts that an extension of a CSM to a QSM satisfying these constraints always exists.

\begin{theorem}\label{thm: quantum extension}
    Every PSM $\langle \M,P(\bu)\rangle$, consisting of a CSM $\M = \langle \bU,\bV,\bF\rangle$ and a probability distribution $P(\bu)$ over exogenous variables, can be extended to a QSM $\Q=\langle (\bV'',\bm{\Lambda}'', S''), \rho_{\bV''S''|\bV''\bm{\Lambda}''}^W,\{\tau^{u_i}_{\Lambda''_i}\}_{u_i} \rangle$ such that
    \begin{align}
        P(\bv)
        &= \sum_{u_i}\prod_{i=1}^n \delta_{v_i,f_i(pa_i,u_i)} P(u_i) \label{eq: causal Markov condition for M}\\
        &= \sum_\bu\tr_{S''^{\mathrm{in}}\bV''}[\rho_{\bV''S''\mid \bV''\bm{\Lambda}''}^W 
        (\tau^\mathrm{v_1}_{V''_1} \otimes \cdots \otimes \tau^\mathrm{v_n}_{V''_n}) \otimes (\widetilde{\tau}^{u_1}_{\Lambda''_1} \otimes \cdots \otimes \widetilde{\tau}^{u_n}_{\Lambda''_n})]\; ,\label{eq: causal Markov condition for M II}
    \end{align}
    In particular, $\Q$ preserves the independence conditions between variables $\bV$ in $\M$ (as defined by $\bF$),
    \begin{equation}\label{eq: factorisation of \Q}
        \rho_{\bV''S''\mid \bV''\bm{\Lambda}''}^W = \prod_{i=1}^n \rho_{V''_iS''_i\mid Pa''_i\Lambda''_i}^{W_i}\; .
    \end{equation}
\end{theorem}

\begin{proof} \textit{(Sketch)} The proof consists of several parts:
    \begin{itemize}
        \item[(i)] we find a binary extension of the CSM $\M$,
        \item[(ii)] we extend the binary CSM to a binary, reversible CSM, where all functional relations are bijective, 
        \item[(iii)] we encode classical copy operations in a QSM using CNOT-gates,
        \item[(iv)] by promoting classical variables to quantum nodes, and by linearly extending bijective functions between classical variables to isometries, we construct a QSM $\Q$, which extends the PSM $\langle \M,P(\bu) \rangle$ as desired.
    \end{itemize}

For details of the proof, see App.~\ref{app: proof quantum extension}.
\end{proof}

We will see in Sec.~\ref{sec:analysis} that a QSM admits different types of counterfactual queries, some of which are genuinely quantum, that is, they do not arise in a CSM. Nevertheless, Thm.~\ref{thm: quantum extension} implies that counterfactual queries arising in a PSM $\langle \M,P(\bu)\rangle$ coincide with the corresponding queries in its quantum extension $\Q$.

\begin{corollary}\label{cor: lifting classical counterfactuals}
    The evaluation of a counterfactual in a (PSM) $\langle \M,P(\bu)\rangle$ coincides with the evaluation of the corresponding do-interventional counterfactual (see also Sec.~\ref{sec:analysis}) in its quantum extension $\Q$.
\end{corollary}

\begin{proof}
    Given a distribution over exogenous nodes, Thm.~\ref{thm: quantum extension} assures that do-interventions in Eq.~(\ref{eq: do-intervention}) yield the same prediction---whether evaluated via Eq.~(\ref{eq: prediction step in CSM}) in $\M$ or as a do-interventional counterfactual via Eq.~(\ref{eq: prediction step in QSM}) in $\Q$. This leaves us with the update step in Pearl's analysis of counterfactuals (cf. Thm.~\ref{thm: evaluation of classical counterfactuals}). More precisely, we need to show that the Bayesian update in Eq.~(\ref{eq: Bayesian update in QSM}) does not affect the distribution over the space of additional ancillae $\bT'$ and $\bm{\Lambda}'$ in the proof of Thm.~\ref{thm: quantum extension}. This is a simple consequence of the way distributions $P(\bu)$ over exogenous nodes in $\M$ are encoded in $\Q$.
    
    First, the distribution over copy ancillae $\Lambda'_i$ is given by a $\delta$-distribution peaked on the state $|0\rangle\langle 0|_{\Lambda_i}$ (see Eq.(\ref{eq: encoding of distributions over exogenous nodes in M}) in App.~\ref{app: proof quantum extension}). In other words, we have full knowledge of the initialization of the copy ancillae, hence, the update step in Eq.~(\ref{eq: Bayesian update in QSM}) is trivial in this case.
    
    Second, let $P(\bu') = P(\bu,\bt')$ be any distribution over exogenous nodes in the binary, reversible extension $\M'$ of $\M$ (see (i) and (ii) in App.~\ref{app: proof quantum extension}) such that $P(\bu) = \sum_{\bt' \in \bT'} P(\bu,\bt')$, that is, $P(\bu)$ arises from $P(\bu')$ by marginalisation under the discarding operation $\pi$ (see (ii) in App.~\ref{app: proof quantum extension}).\footnote{A canonical choice for $P(\bu')$ is the product distribution of $P(\bu)$ and the uniform distribution over $\bT'$, $P(\bu') = \frac{1}{|\bT'|} P(\bu)$.} But since the variables $T'_i$ in $U'_i = T'_i \times U_i$ are related only to the sink node $S'_i$ via $f'_i$ (see Eq.~(\ref{eq: reversible function extension}) in App.~\ref{app: proof quantum extension}), we have $P_\bz(\ba|\bu') = P_\bz(\ba|\bu,\bt') = P_\bz(\ba|\bu)$.
    The marginalised updated distribution thus reads
    \begin{equation}
        \sum_{\bt'\in\bT'} P_\bz(\bu,\bt'|\ba)
        = \sum_{\bt'\in\bT'} \frac{P_\bz(\ba|\bu,\bt')P(\bu,\bt')}{P_\bz(\ba)}
        = \sum_{\bt'\in\bT'} \frac{P_\bz(\ba|\bu)P(\bu,\bt')}{P_\bz(\ba)}
        = \frac{P_\bz(\ba|\bu)P(\bu)}{P_\bz(\ba)}
        = P_\bz(\bu|\ba)\; .
    \end{equation}
    In other words, Bayesian inference in Eq.~(\ref{eq: Bayesian update in QSM}) commutes with marginalisation.
\end{proof}

Thm.~\ref{thm: quantum extension} and Cor.~\ref{cor: lifting classical counterfactuals} show that our definition of QSMs in Def.~\ref{def: quantum structural causal model} generalizes that of CSMs in Def.~\ref{def: classical_structural_cm_def}. What is more, this generalization is proper: a QSM cannot generally be thought of as a CSM, while also keeping the relevant independence conditions between the variables of the model. Indeed, casting a QSM to a CSM is to specify a local hidden variable model for the QSM, yet a general QSM will not admit a local hidden variable model.

In short, the \emph{counterfactual probabilities} defined by a QSM can generally not be interpreted as \emph{probabilities of counterfactuals} (to be true). In \red Sec.~\ref{sec:analysis}, we will further analyse the distinctions between classical and quantum counterfactuals, and \blk see some instances of counterfactual queries in the quantum case that do not have an analog in the classical case.

\section{Interpretation of \red quantum \blk counterfactual queries}\label{sec:analysis}

In this section, we emphasize some crucial differences between the semantics of counterfactuals in classical and quantum causal models. Recall that in order to compute the probability of a counterfactual in a classical structural causal model (CSM), a do-intervention has to be considered in at least one of the nodes. Indeed, there is no way for the antecedent of the counterfactual query to be true \red without some modification in the model, \blk since a complete specification of the values of exogenous variables determines the values of endogenous variables, and thus determines the antecedent to have its actual value. CSMs are inherently deterministic.

In contrast, in a quantum structural causal model (QSM) the probability \red that a different outcome would have been obtained \blk can be nonzero even without a do-intervention, since even maximal knowledge of the events at the exogenous nodes does not, in general, determine the outcomes of endogenous instruments. QSMs are inherently probabilistic. 

As a consequence, we will distinguish between two kinds of counterfactuals in the quantum case, namely, \emph{passive} and \emph{active} counterfactuals, which we define and discuss examples of in Sec.~\ref{sec: different counterfactuals for QSM}. In Sec.~\ref{sec: principle of minimality}, we provide an argument for the disambiguation between passive and active counterfactuals, when faced with an ambiguous (classical) counterfactual query. Moreover, as a consequence of the richer semantics of quantum counterfactuals, in Sec.~\ref{sec: Bell scenario} we show how (passive) quantum counterfactuals break the \red link between \blk causal and counterfactual dependence \red that exists \blk in the classical setting. We discuss this explicitly in the case of the Bell scenario.

\subsection{Passive and active counterfactuals}\label{sec: different counterfactuals for QSM}
In Sec.~\ref{sec:cf_q_evaluation} we outlined a three-step procedure to evaluate counterfactual probabilities in quantum causal models. Note that, unlike in its classical counterpart (Thm.~\ref{thm: evaluation of classical counterfactuals}), an arrow-breaking do-intervention is not necessary in order to make the antecedent of the counterfactual true. Counterfactual queries can therefore be evaluated without a do-intervention on the underlying causal graph, and, in particular, without changing the instruments performed at quantum nodes at all. Indeed, \red according to Def.~\ref{def: quantum counterfactual query}, the expected counterfactual probability has a well-defined numerical value whenever the antecedent has a nonzero probability $P_{\bz'}^{\bm{\lambda}}(\bb')$ of occurring for all values of the background variables $\bm{\lambda}$ that are compatible with the evidence $\ba|\bz$, that is,
\begin{equation}\label{eq: condition passive counterfactual}
    \forall \bm{\lambda}\in \bm{\Lambda} \; \left(P_{\bz}(\bm{\lambda}|\ba) > 0 \implies
    P_{\bz'}^{\bm{\lambda}}(\bb') > 0\right)\,.
\end{equation}
It is possible for Eq.~\eqref{eq: condition passive counterfactual} to be satisfied even while keeping the endogenous instruments fixed, i.e.~even if $\bz'=\bz$. \blk 
Crucially, unlike in the classical case, we will see that this may be the case even if the antecedent $\bb'|\bz'$ is incompatible with the observed values $\ba|\bz$. This motivates the following distinction for quantum counterfactuals.

\begin{definition}\label{def: types of quantum counterfactuals}
    Let $\Q=\langle (\bA, \bL, S), \rho_{\bA S|\bA \bL}^U,\{\tau^{\lambda_i}_{\Lambda_i}\}_{\lambda_i} \rangle$ be a quantum structural causal model. A counterfactual query (following Def.~\ref{def: quantum counterfactual query}) is called a \emph{passive counterfactual} if \red  $\bz'_\bB \equiv \bz'|_\bB = \bz|_\bB \equiv \bz_\bB$, \blk that is, if no intervention is performed on the nodes specified by the antecedent; otherwise it is called an \emph{active counterfactual}. 
    
    The special case of an active counterfactual where $\bz'$ specifies a do-intervention, $\tau_A^{\mathrm{do}(\bm{\rho})} = \{(\bm{\rho}_{\bA^{\out}})^T \otimes \mathbb{I}_{\bA^{\ins}}\}$ (see Eq.~(\ref{eq: quantum do-intervention})), will also be called a \emph{do-interventional counterfactual}.
\end{definition}

In the following, we discuss two examples of passive, active and do-interventional counterfactuals.

\begin{figure}
		\centering
		\begin{tikzpicture}[node distance={15mm}, thick, main/.style = {draw, circle}] 
                \node[main] (1)  {$B$};
                \node[main] (2) [below of=1] {$A$};
                \node[main] (3) [below=0.5cm of 2]{$\Lambda$}; 
			  \node[] (4) [ left=0.5cm of 1] {$z_B$}; 
                \node[] (5) [ left=0.5cm of 2] {$z_A$}; 
                \node[] (6) [ right=0.5cm of 1] {$b$}; 
                \node[] (7) [ right=0.5cm of 2] {$a$};
               \node[] (9) [ right=0.5cm of 3] {$\lambda$};
			\draw[dashed, ->] (3) -- (2); 
 			\draw[->] (2) -- (1);
                \draw[->] (4) -- (1);
                \draw[->] (5) -- (2);
                \draw[->] (1) -- (6);
                \draw[->] (2) -- (7);
               \draw[->] (3) -- (9);
 		\end{tikzpicture}
		\caption{A graphical representation of a quantum causal model with endogenous nodes $\bm{A} = \{A,B\}$, and exogenous node $\Lambda$. Here we also represent the choice of instruments as inputs to a node, and their corresponding outcomes as outputs.} \label{cf_example}
\end{figure}
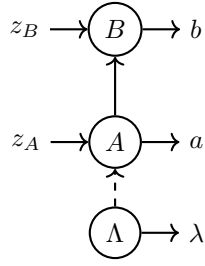

\begin{example}\label{ex: passive vs active CF}
Consider the causal graph in Fig.~\ref{cf_example} and a compatible QSM $M_Q = \langle (\bA,\Lambda), \rho^U_{AB|A\Lambda},\{\tau^\lambda\}_\lambda\rangle$, where $\bA = \{A,B\}$ represent endogenous nodes, and $\Lambda$ represents an exogenous node with the following discard-and-prepare instrument,
\begin{equation}\label{eq: preparation instrument 1}
    \{\tau^\lambda\}_{\lambda=0,1} = \Bigl\{\frac{1}{2}(([0]_{\Lambda^{\mathrm{out}}})^T \otimes \mathbb{I}_{\Lambda^{\mathrm{in}}}), \frac{1}{2}(([1]_{\Lambda^{\mathrm{out}}})^T\otimes \mathbb{I}_{\Lambda^{\mathrm{in}}}) \Bigr\}\; ,
\end{equation}
such that $\tau_\Lambda^{\frac{1}{2}\mathbb{I}} = \sum_{\lambda=0,1} \tau^\lambda_\Lambda$ prepares the maximally mixed state, and we assume identity channels between pairs of nodes,
\begin{equation}\label{eq: example identity process operator}
    \rho_{AB|A \Lambda}^U = \rho_{B|A}^{\mathrm{id}} \rho_{A|\Lambda}^{\mathrm{id}} = \rho_{B^{\ins}|A^{\out}}^{\mathrm{id}} \rho_{A^{\ins}|\Lambda^{\out}}^{\mathrm{id}}\; .
\end{equation}
With respect to the model $M_Q$, we will calculate expected counterfactual probabilities of the form 
\pur{$P^{a=+|\bz}_{a'=-|\bz'}(b')$}, \blk  where we fix the actual instruments $\bz=(z_A=1,z_B)$ at endogenous nodes with
\begin{equation}
     \mathcal{I}_A^{z_A=1} = \{ ([+]_{A^{\mathrm{out}}})^T \otimes [+]_{A^{\mathrm{in}}}, ([-]_{A^{\mathrm{out}}})^T \otimes [-]_{A^{\mathrm{in}}} \}\; , \quad \quad \quad \quad \quad \mathcal{I}_B^{z_B} = \{\tau_B^{b|z_B}\}_b\; ,
\end{equation}
but consider different counterfactual instruments $\bz'=(z'_A,z'_B)$, corresponding to (i) passive, (ii) do-interventional, and (iii) active counterfactual queries. 
To this end, we first calculate the conditional process operators (cf. Eq.~(\ref{eq: conditional process operator})), conditioned on outcomes $\lambda \in \{0,1\}$ of the instrument in Eq.~(\ref{eq: preparation instrument 1}):
\begin{align}\label{eq: l1 conditional process operator}
    \sigma_{AB}^{\lambda=0}
    &= \frac{\tr_\Lambda\left[\rho_{AB|A\Lambda}^U \widetilde{\tau}^{\lambda=0}_{\Lambda}\right]}{P(\lambda=0)}
    = \rho_{B^{\mathrm{in}}|A^{\mathrm{out}}}^{\mathrm{id}} \otimes [0]_{A^{\mathrm{in}}}\; ,\nonumber \\
    \sigma_{AB}^{\lambda=1}
    &= \frac{\tr_\Lambda\left[\rho_{AB|A\Lambda}^U \widetilde{\tau}^{\lambda=1}_{\Lambda}\right]}{P(\lambda=1)}
    = \rho_{B^{\mathrm{in}}|A^{\mathrm{out}}}^{\mathrm{id}} \otimes [1]_{A^{\mathrm{in}}}\; .
\end{align}

\begin{itemize}
    \item[(i)] Passive case: \emph{``Given that $a=+$ occurred in the actually performed instrument $\mc{I}^1_A$, what is the probability that $b'$ would have obtained using the instrument $\mc{I}^{z'_B=1}_B$, had it been that $a'=-$, using the instrument $\mc{I}^1_A$?''}. 
    
\red
In the abduction step, we update our information about the exogenous node, given the evidence. \pur Note that observing the outcome $a=+$ at node $A$ in this particular case gives us no information about the outcome at the exogenous node since both outcomes at $A$ occur with equal probability for both of the possible values of the exogenous variable.  This is expressed in the following conditional probabilities (c.f. Eq.~\eqref{eq: Bayesian update in QSM} and the denominator of Eq.~\eqref{eq: prediction cf prob}), noting that here the counterfactual antecedent is $a'$, in place of $b'$ in Eq.~\eqref{eq: prediction cf prob}:
\begin{align}\label{eq:ex1 abduction_1}
    P_\bz(\lambda=0|a=+) &= \frac{1}{2} \quad \wedge \quad P_{\bz'}^{\lambda=0}(a' = -) =\frac{1}{2}\; ,\\
    \label{eq:ex1 abduction_2}
    P_\bz(\lambda=1|a=+) &= \frac{1}{2} \quad \wedge \quad P_{\bz'}^{\lambda=1}(a' = -) =\frac{1}{2}\; .
\end{align}
\red From the above, we see that we satisfy the conditions for a passive counterfactual to have a well-defined numerical value, given in Eq.~\eqref{eq: condition passive counterfactual}. \pur Since $z'_A = z_A = 1$, this is a passive counterfactual, hence, no action step, that is, no intervention is needed.
For the prediction step, we first compute the required counterfactual probabilities from Eq.~\eqref{eq: prediction cf prob},
\begin{align}
    P_{\bz'}^{\lambda}(b'|a'=-) &= \frac{P_{\bz'}^{\lambda}(b',a'=-)}{P_{\bz'}^{\lambda}(a'=-)}
    = \dfrac{\mathrm{Tr}_{AB}[\sigma_{AB}^{\lambda}(\tau_A^{a'=-|1}\otimes\tau_B^{b'|z'_B=1})]}{\mathrm{Tr}_{AB}[\sigma_{AB}^{\lambda}(\tau_A^{a'=-|1}\otimes\tau_B^{|z'_B=1})]} \; .
\end{align}
Using Eq.~\eqref{eq: l1 conditional process operator}, the counterfactual probability for the different values of $\lambda$ can be written as
\begin{align}\label{eq:ex 1 CFP l0}
    P_{\bz'}^{\lambda=0}(b'|a'=-) 
    &= 2\;\mathrm{Tr}_{AB}[\rho_{B^{\mathrm{in}}|A^{\mathrm{out}}}^{\mathrm{id}} \otimes [0]_{A^{\mathrm{in}}}(\tau_A^{a'=-|1}\otimes\tau_B^{b'|z'_B=1})]\;,
\end{align}
\begin{align}\label{eq:ex 1 CFP l1}
    P_{\bz'}^{\lambda=1}(b'|a'=-) 
    &= 2\;\mathrm{Tr}_{AB}[\rho_{B^{\mathrm{in}}|A^{\mathrm{out}}}^{\mathrm{id}} \otimes [1]_{A^{\mathrm{in}}}(\tau_A^{a'=-|1}\otimes\tau_B^{b'|z'_B=1})]\;.
\end{align}

The expected counterfactual probability (refer Eq.\eqref{eq: prediction step in QSM},
\begin{align}
    P_{a'=-|z_A'=1}^{a=+|\bz}(b') 
    &= \sum_{\lambda \in \Lambda} P_{\bz}(\lambda|a=+) P_{\bz'}^{\lambda}(b'|a'=-)\; \\
    &= \mathrm{Tr}_{AB}[\rho_{B^{\mathrm{in}}|A^{\mathrm{out}}}^{\mathrm{id}} \otimes [0]_{A^{\mathrm{in}}}(\tau_A^{a'=-|1}\otimes\tau_B^{b'|z'_B=1}) + \rho_{B^{\mathrm{in}}|A^{\mathrm{out}}}^{\mathrm{id}} \otimes [1]_{A^{\mathrm{in}}}(\tau_A^{a'=-|1}\otimes\tau_B^{b'|z'_B=1})]\\
    &= \mathrm{Tr}_{AB}[\rho_{B^{\mathrm{in}}|A^{\mathrm{out}}}^{\mathrm{id}} \otimes \mathbb{I}_{A^{\mathrm{in}}}([-]_{A^{\mathrm{out}}})^T \otimes [-]_{A^{\mathrm{in}}}\otimes\tau_B^{b'|z'_B=1})]\\
    &=\mathrm{Tr}_B[[-]_{B^{\mathrm{in}}}\tau_B^{b'|z'_B=1}] \;.
\end{align}
The obtained expected counterfactual probability is \red thus numerically equivalent to the probability of obtaining outcome $b'$ for the instrument specified by $z_B'=1$, given the preparation of state $[-]$ at the input of node $B$. In other words, the counterfactual probability is simply dictated by the counterfactual outcome at $A$, which is a sensible result, since in this case the actual outcome $a$ gives us no information about the exogenous variables $\lambda$, the counterfactual outcome is possible for either of the values of $\lambda$, and furthermore the effect of the exogenous variables on $B$ is screened off by the outcome of the projective counterfactual instrument considered for node $A$.

\blk

\item[(ii)] Do-interventional case: \emph{``Given that $a=+$ occurred in the actually performed instrument $\mc{I}^1_A$, what is the probability that $b'$ would have obtained using the instrument $\mc{I}^{z'_B=2}_B$, had it been that $a'=-$, using the instrument $\tau^{\mathrm{do}([-])}_A$?''}.

Here, instead of \red $\mc{I}^{z'_A=1}_A$, we perform the single-element instrument corresponding to a do-intervention,
\begin{equation}\label{eq: do-example}
    \mc{I}^{z'_A=2}_A = \{\tau_A^{a'=-|2}\}=\{\tau_A^{\mathrm{do}([-])}\} = \{([-]_{A^{\mathrm{out}}})^T \otimes \mathbb{I}_{A^{\mathrm{in}}}\; \},
\end{equation} \blk
which discards the input and prepares the state $[-]$ at the output of A. \red As before, the actual instruments are denoted by $\bz=(z_A=1,z_B)$. Since the actual instruments and evidence are the same as in case (i), the abduction step yields the same result, given in the left side of Eqs.~\eqref{eq:ex1 abduction_1} and \eqref{eq:ex1 abduction_2}. In the action step, however, we modify the instrument at node $A$, and the counterfactual instruments are now denoted by $\bz' = (z'_A=2,z'_B)$. 

For the prediction step, we again first compute the required counterfactual probabilities from Eq.~\eqref{eq: prediction cf prob},
\pur
\begin{align}
    P_{\bz'}^{\lambda}(b'|a'=-) &= \frac{P_{\bz'}^{\lambda}(b',a'=-)}{P_{\bz'}^{\lambda}(a'=-)}
    = \dfrac{\mathrm{Tr}_{AB}[\sigma_{AB}^{\lambda}(\tau_A^{a'=-|2}\otimes\tau_B^{b'|z'_B=2})]}{\mathrm{Tr}_{AB}[\sigma_{AB}^{\lambda}(\tau_A^{a'=-|2}\otimes\tau_B^{|z'_B=2})]} \; .
\end{align}
\red The counterfactual probabilities for the different values of $\lambda$ thus take the same form as in Eqs.~\eqref{eq:ex 1 CFP l0} and \eqref{eq:ex 1 CFP l1}, but with the instrument element $\tau_A^{a'=-|2}$ in place of $\tau_A^{a'=-|1}$. The expected counterfactual probability corresponding to query (ii) can then be computed as
\begin{align}
    P_{a'=-|z_A'=2}^{a|\bz}(b') 
    &= \sum_{\lambda \in \Lambda} P_{\bz}(\lambda|a=+) P_{\bz'}^{\lambda}(b'|a'=-)\; \\
    &= \frac{1}{2}\;\mathrm{Tr}_{AB}[\rho_{B^{\mathrm{in}}|A^{\mathrm{out}}}^{\mathrm{id}} \otimes \mathbb{I}_{A^{\mathrm{in}}}(([-]_{A^{\mathrm{out}}})^T \otimes \mathbb{I}_{A^{\mathrm{in}}}\otimes\tau_B^{b'|z'_B=2})]\\
    &=\mathrm{Tr}_B[[-]_{B^{\mathrm{in}}}\tau_B^{b'|z'_B=2}] \;.
\end{align}
As the instrument $z'_B=2$ is arbitrary, we see that we obtain the same result for the do-interventional as for the passive counterfactual in this case. So we see that although the do-intervention guarantees the conditions for the counterfactual antecedent to have occurred, it is not necessary in situations where it \emph{could} have occurred passively, as in case (i). \blk
\blk
\item[(iii)] Active case: we ask \emph{``Given that $a=+$ occurred in the actually performed instrument $\mc{I}^1_A$, what is the probability that $b'$ would have obtained using the instrument $\mc{I}^{z'_B=3}_B$, had it been that $a'=-$, using the instrument $\mc{I}^3_A$''}.

\red Again, the abduction step yields the same results as in cases (i) and (ii). For the action step, this is an active counterfactual whenever  $\mc{I}^1_A \neq \mc{I}^3_A$. Specifically, let's consider
\begin{equation}\label{eq: active example}
    \mc{I}^3_A = \{([+]_{A^{\mathrm{out}}})^T \otimes [\phi]_{A^{\mathrm{in}}}, ([-]_{A^{\mathrm{out}}})^T \otimes [\overline{\phi}]_{A^{\mathrm{in}}}\}\; ,
\end{equation}
\red where $[{\phi}]$ is an arbitrary projector and $[\overline{\phi}]$ is its orthogonal complement. This instrument performs a projective measurement on the $\{[{\phi}],[\overline{\phi}]\}$ basis on the input of $A$, and prepares the corresponding state $[+]/[-]$ at the output. Considering the counterfactual probabilities from Eq.~\eqref{eq: prediction cf prob}, we find them to have the same value for the two possible values of $\lambda$, largely independent of $[\phi]$,
\begin{align}
    P_{\bz'}^{\lambda}(b'|a'=-) &= \frac{P_{\bz'}^{\lambda}(b',a'=-)}{P_{\bz'}^{\lambda}(a'=-)}
    = \dfrac{\mathrm{Tr}_{AB}[\sigma_{AB}^{\lambda}(\tau_A^{a'=-|3}\otimes\tau_B^{b'|z'_B=3})]}{\mathrm{Tr}_{AB}[\sigma_{AB}^{\lambda}(\tau_A^{a'=-|3}\otimes\tau_B^{|z'_B=3})]} =\mathrm{Tr}_{B}[[-]_{B^{\mathrm{in}}}\otimes\tau_B^{b'|z'_B=3})]\; ,\label{eq: ex_1_active_cf}
\end{align}
 provided the denominator of Eq.~\eqref{eq: ex_1_active_cf} is nonzero (which does depend on $[\phi]$). And since the abducted probabilities for the two values of $\lambda$ in this case are equal, the expected counterfactual probability also has the same numerical value whenever the denominator of Eq.~\eqref{eq: ex_1_active_cf} is nonzero for \emph{both values of $\lambda$}, namely,
\begin{align}\label{eq:ex1.3 ECP}
P_{a'|z_A'=3}^{a|\bz}(b') = \sum_{\lambda \in \Lambda} P_{\bz}(\lambda|a=+) P_{\bz'}^{\lambda}(b'|a'=-) = \mathrm{Tr}_{B}[[-]_{B^{\mathrm{in}}}\otimes\tau_B^{b'|z'_B=3})]\;.
\end{align}
\pur However, there are exceptions to this, for some specific values of $[\phi]$.
The denominator of Eq.~\eqref{eq: ex_1_active_cf} will be zero \red (in other words, the counterfactual antecedent will be impossible) for $\lambda = 0$ when $[\phi] = [0]$  and for $\lambda = 1$ when $[\phi] = [1]$. In these cases we set the value of the corresponding counterfactual probability in Eq.~\eqref{eq: ex_1_active_cf} -- and for the expected counterfactual probability in Eq.~\eqref{eq:ex1.3 ECP} -- to $*$, by definition, to indicate a counterpossible (see Section~\ref{sec:cf_q_evaluation}).
Whenever it is numerically well-defined, on the other hand, the expected counterfactual probability has the same value as in cases (i) and (ii), considering that the instrument denoted by $z'_B=3$ is again arbitrary.
\blk
\end{itemize}
\end{example}

\begin{example}\label{ex: passive vs active CF 2}
    Consider the same setup as in Ex.~\ref{ex: passive vs active CF}, but with a different instrument at the exogenous node $\Lambda$,
    \begin{equation}\label{eq: preparation instrument 2}
        \mathcal{I}_{\Lambda}^2
        = \{\tau^\lambda\}_{\lambda=+,-}
        = \Bigl\{\frac{1}{2}(([+]_{\Lambda^{\mathrm{out}}})^T \otimes \mathbb{I}_{\Lambda^{\mathrm{in}}}), \frac{1}{2}(([-]_{\Lambda^{\mathrm{out}}})^T \otimes \mathbb{I}_{\Lambda^{\mathrm{in}}}) \Bigr\}\; ,
    \end{equation}
    \red This instrument also prepares at the output of the exogenous node, on average, the maximally entangled state, that is, $\tau_\Lambda^{\frac{1}{2}\mathbb{I}} = \sum_{\lambda=+,-} \tau^\lambda_\Lambda$. This implies that, averaging over the values of the exogenous variable, we obtain the same process operator $\sigma_{AB}$ for the endogenous nodes, and thus the same quantum causal model as per Def.~\ref{def: qmc}. However, it implies a \emph{distinct} quantum structural causal model, as per Def.~\ref{def: quantum structural causal model}. Let's analyse some of the consequences of this fact for the evaluation of counterfactuals. \blk
    
    In contrast to Ex.~\ref{ex: passive vs active CF}, the conditional process operators for outcomes $\lambda \in \{+,-\}$ now become
    \begin{align}\label{eq: ex2 conditional process operator}
        \sigma_{AB}^{\lambda=+}
        &= \frac{\tr_\Lambda\left[\rho_{AB|A\Lambda}^U \widetilde{\tau}^{\lambda=+}_{\Lambda}\right]}{P(\lambda=+)}
        = \rho_{B^{\mathrm{in}}|A^{\mathrm{out}}}^{\mathrm{id}} \otimes [+]_{A^{\mathrm{in}}}\; ,\\
        \sigma_{AB}^{\lambda=-}
        &= \frac{\tr_\Lambda\left[\rho_{AB|A\Lambda}^U \widetilde{\tau}^{\lambda=-}_{\Lambda}\right]}{P(\lambda=-)}
        = \rho_{B^{\mathrm{in}}|A^{\mathrm{out}}}^{\mathrm{id}} \otimes [-]_{A^{\mathrm{in}}}\; .
    \end{align}
Again, we compute the (i) passive, (ii) do-interventional, and (iii) active counterfactual probabilities for the same counterfactual queries as in Ex.~\ref{ex: passive vs active CF}.

\red The required abducted probabilities (Eq.~\eqref{eq: Bayesian update in QSM}) will be again the same for all queries -- which share the same actual instruments and evidence -- and are calculated to be \blk 
\begin{align}
    P_{\bz}(\lambda=+|a=+) &= 1\; ,\label{eq3: abduction ex 2 plus}\\
    P_{\bz}(\lambda=-|a=+) &= 0\; .\label{eq3: abduction ex 2 minus}
\end{align}
    \begin{itemize}
        \item[(i)] Passive case: \emph{``Given that $a=+$ occurred in the actually performed instrument $\mc{I}^1_A$, what is the probability that $b'$ would have obtained using the instrument $\mc{I}^{z'_B=1}_B$, had it been that $a'=-$, using the instrument $\mc{I}^1_A$?''} 

\red 
Now the antecedent becomes a counterpossible for $\lambda = +$, and the corresponding counterfactual probability is set to 
\begin{align}
    P_{\bz'}^{\lambda=+}(b'|a'=-) &= \frac{P_{\bz'}^{\lambda=+}(b',a'=-)}{P_{\bz'}^{\lambda=+}(a'=-)}= *\;.
\end{align}

And since the abducted probability for $\lambda = +$ is nonzero, the expected counterfactual probability is also by definition,
\begin{align}
    P_{a'=-|z_A'=1}^{a=+|\bz}(b') 
    &= \sum_{\lambda \in \Lambda} P_{\bz}(\lambda|a=+) P_{\bz'}^{\lambda}(b'|a'=-) = *\;.
\end{align}
We thus see that the change in the preparation instrument at the exogenous node leads to a different result from Ex.~\ref{ex: passive vs active CF}  -- even though both cases lead to the same process operator for the endogenous nodes, when averaged over the background variables. This illustrates the dependence of counterfactual questions on the quantum structural causal model rather than on the quantum causal model for the endogenous nodes alone.

\blk
        \item[(ii)] Do-interventional case: \emph{``Given that $a=+$ occurred in the actually performed instrument $\mc{I}^1_A$, what is the probability that $b'$ would have obtained using the instrument $\mc{I}^{z'_B=2}_B$, had it been that $a'=-$, using the instrument $\mc{I}^{z'_A=2}_A=\tau^{\mathrm{do}([-])}_A$?''}
        
\red The do-intervention Eq.~(\ref{eq: do-example}) now guarantees that the denominator of the counterfactual probability is nonzero (indeed 1),
\begin{align}
    P_{\bz'}^{\lambda}(b'|a'=-) = \frac{P_{\bz'}^{\lambda}(b',a'=-)}{P_{\bz'}^{\lambda}(a'=-)}=P_{\bz'}^{\lambda}(b',a'=-)\;.
\end{align}

Since the abducted probability for $\lambda=-$, given the observed evidence, is zero, the expected counterfactual probability (Eq.~\eqref{eq: prediction step in QSM}) is equal to the counterfactual probability corresponding to $\lambda = +$, which is computed as 
\begin{align}
    P_{a'=-|z_A'=2}^{a=+|\bz}(b') 
    & = P_{\bz'}^{\lambda=+}(b'|a'=-)\\
    &= \mathrm{Tr}_{AB}[\rho_{B^{\mathrm{in}}|A^{\mathrm{out}}}^{\mathrm{id}} \otimes [+]_{A^{\mathrm{in}}}(\tau_A^{a'=-|2}\otimes\tau_B^{b'|z'_B=2})]\\
    &= \mathrm{Tr}_{B}[[-]_{B^{\mathrm{in}}}\tau_B^{b'|z'_B=2}]\; .
\end{align}
Note that this is the same result as for the same query in Example 1, since the do-intervention breaks the causal dependence from the exogenous variables in this particular case.
\blk
    
\item[(iii)] Active case: \emph{``Given that $a=+$ occurred in the actually performed instrument $\mc{I}^1_A$, what is the probability that $b'$ would have obtained using the instrument $\mc{I}^{z'_B=3}_B$, had it been that $a'=-$, using the instrument $\mc{I}^3_A$?''}

Using the instrument in Eq.~(\ref{eq: active example}), we \red find the counterfactual probability to have the same form for the two values of $\lambda$, and largely independent of $[\phi]$ (provided the denominator is nonzero, which again depends on $[\phi]$), 
\begin{align}
    P_{\bz'}^{\lambda}(b'|a'=-) &= \frac{P_{\bz'}^{\lambda}(b',a'=-)}{P_{\bz'}^{\lambda}(a'=-)}=\mathrm{Tr}_{B}[[-]_{B^{\mathrm{in}}}\otimes\tau_B^{b'|z'_B=3})]\;.\label{ex_2_active_cf}
\end{align}

Distinct from Example 1, however, the denominator of Eq.~\eqref{ex_2_active_cf} will be zero (in other words, the counterfactual antecedent will be impossible) for $\lambda = +$ when $[\phi] = [+]$  and for $\lambda = -$ when $[\phi] = [-]$. However, since the only nonzero abducted probability is $P_{\bz}(\lambda=+|a=+) = 1$, the expected counterfactual probability is numerically well-defined for all $[\phi] \neq [+]$,
\begin{align}
    P_{a'=-|z_A'=3}^{a=+|\bz}(b') 
    = P_{\bz'}^{\lambda=+}(b'|a'=-)
    = \mathrm{Tr}_{B}[[-]_{B^{\ins}}\tau_B^{b'|z'_B=3}]\; .
\end{align}

    \end{itemize}
\end{example}

\red 
Comparing the two examples, we see that while the passive counterfactual in Ex.~\ref{ex: passive vs active CF 2} is always a counterpossible, that is, it has an impossible antecedent (and is thus assigned a conventional value $*$), the same passive counterfactual evaluated in Ex.~\ref{ex: passive vs active CF} yields an expected counterfactual probability that is always numerically well-defined. This is in stark contrast to the classical case, where - as a consequence of the intrinsic determinism of classical structural causal models - a passive interpretation of a counterfactual query would always result in a counterpossible. Note also that both examples have the same average state (a maximally mixed state) prepared at the exogenous node, showing that different contexts for the state preparations of the same mixed state, and thus different quantum structural models, can result in different evaluations for a quantum counterfactual. We also see that in both examples the do-interventional counterfactual is numerically well-defined and has the same value, whereas the counterfactual probabilities in active cases are counterpossibles for different counterfactual instruments.

These distinctions then lead to the question: what should we do when a counterfactual statement is not already in standard form, and thus it is unclear whether it should be interpreted as a passive or active counterfactual? \blk

\subsection{Disambiguation of counterfactual queries: the principle of minimality}\label{sec: principle of minimality}

Note that classical counterfactuals evaluated with respect to a probabilistic structural causal model $\langle M,P(\bu) \rangle$ correspond to do-interventional (quantum) counterfactual queries when evaluated with respect to the quantum extension $M_Q$ (cf. Cor.~\ref{cor: lifting classical counterfactuals}). In fact, a classical counterfactual query in Def.~\ref{def: counterfactual} is always defined in terms of a do-intervention, since this is the only way to make the antecedent true. In this sense, we may say that classical counterfactual queries naturally embed into our formalism as do-interventional counterfactuals.

Yet, the richer structure of quantum counterfactuals, as seen in the previous section, may sometimes allow for a different interpretation of a classical counterfactual query, in particular, the antecedent of a quantum counterfactual can sometimes be true without an intervention. This leaves a certain ambiguity if we want to interpret a classical counterfactual as a quantum counterfactual query according to Def.~\ref{def: quantum counterfactual query}: for the latter, one must specify a counterfactual instrument, in particular, one must decide whether to interpret the classical counterfactual query passively or actively (do-interventionally).
For example, again referring to the scenario represented in \red Fig.~\ref{cf_example}, consider the query: 

\begin{quote}
\emph{Given that $a=+$, what is the probability that $b'$, had it been that $a'=-$?}
\end{quote}

Note that \emph{all} of the counterfactuals in the previous section are of this form, until we specify what instruments those outcomes correspond to. And in the model of Example 1, both the passive and do-interventional interpretations of this query are always numerically well-defined. \blk

This ambiguity does not occur in a classical structural causal model (CSM), since in that case all the variables are determined by a complete specification of the exogenous variables. Consequently, the only way the antecedent of a counterfactual like the one above could be realized while keeping the background variables fixed, is via some modification of the model.\footnote{We remark that, contrary to Pearl, a counterfactual may also be interpreted as a \emph{backtracking counterfactual}, where the background conditions are not necessarily kept fixed. A semantics for backtracking counterfactuals within a classical SCM has recently been proposed in Ref.~\cite{vonKugelgen2022}.} Pearl justifies the do-intervention as ``the minimal change (to a model) necessary for establishing the antecedent'' \cite{Pearl2000}. \red In our case, due to the split-node structure, a do-intervention is reflected not as a change in the model itself, but as a change in the instrument used at the antecedent nodes, that is, via the use of a do-instrument. \blk

To decide whether a \red counterfactual query not in standard form \blk should be analyzed as passive or active when interpreted with respect to a QSM, we thus propose a \emph{principle of minimality}, motivated by the minimal changes from actuality required in Pearl's analysis. If the antecedent of a counterfactual can be established with \emph{no} change to \red (the instruments applied to) a model -- that is, \blk as in a passive reading of the counterfactual -- this is by definition the \emph{minimal} change.

\begin{definition}[Principle of Minimality]\label{def: Principle of minimality}
    Whenever it is ambiguous whether a counterfactual query should be interpreted as a passive or active counterfactual in a QSM (as by Def.~\ref{def: types of quantum counterfactuals}), it should be interpreted passively if it is not a counterpossible, that is, if its antecedent is not impossible (as by Eq.~(\ref{eq: condition passive counterfactual})).
\end{definition}

Lewis' account of counterfactuals invokes a notion of \emph{similarity} among possible worlds \cite{lewis1973counterfactuals}. For Lewis, one should order the closest possible worlds by some measure of similarity, based on which a counterfactual is declared true in a world $w$ if the consequent of the counterfactual is true in all the closest worlds to $w$ where the antecedent of the counterfactual is true. Arguably, a world in which both the model and instruments are the same, but where the counterfactual antecedent occurs, is closer to the actual world than any world where a different instrument is used. Thus the Principle of Minimality is also justified by a form of Lewis' analysis applied to our case.

In Example 1, however, where both the active and do-interventional readings are numerically well-defined, they also produce the same (expected) counterfactual probability, rendering the ambiguity essentially irrelevant. We next turn to an important example where this is not the case.

\subsection{Causal dependence and counterfactual dependence in the Bell scenario}\label{sec: Bell scenario}

A conceptually important consequence of our semantics for counterfactuals (and relevant for the disambiguation of passive from active counterfactual queries) is that, unlike in the case of Pearl's framework, counterfactual dependence does not in general imply causal dependence. We establish this claim using the pertinent example of a common cause scenario, as shown in Fig.~\ref{cf_example_Bell}. \red This is essentially the causal structure of a Bell scenario, as in Fig.~\ref{Bell}, although here we omit the nodes associated with the choices of setting $X$ and $Y$, which are now the choices of instrument for the quantum nodes $A$ and $B$, and left implicit\footnote{Similarly, the labels $A$ and $B$ of the quantum nodes themselves are not to be confused with the outcomes of instruments that may be used at those nodes}. \blk 
\begin{figure}[b]
	\centering
	\begin{tikzpicture}[node distance={20mm}, thick, main/.style = {draw, circle}] 
			\node[main] (1) {$C$}; 
			\node[main] (2) [ above left of=1] {$A$}; 
			\node[main] (3) [ above right of=1] {$B$};
			\draw[->] (1) -- (2); 
 			\draw[->] (1) -- (3);
 		\end{tikzpicture}
		\caption{Causal graph of a quantum causal model where $C$ is a common cause of $A$ and $B$ (equivalent to the Bell scenario in (Fig.~\ref{Bell}), but with choices of instruments at the quantum nodes $A$ and $B$ left implicit). } \label{cf_example_Bell}
\end{figure}
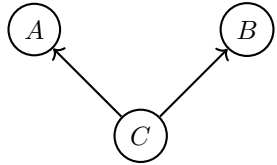

\begin{example}[Bell scenario]
    Consider the causal scenario in Fig.~\ref{cf_example_Bell}, with instruments
    \begin{align}
        \mathcal{I}_{A} &= \{ ([0]_{A^{\mathrm{out}}})^T \otimes [0]_{A^{\mathrm{in}}}, ([1]_{A^{\mathrm{out}}})^T \otimes [1]_{A^{\mathrm{in}}} \}\; ,\\
        \mathcal{I}_{B} &= \{([0]_{B^{\mathrm{out}}})^T \otimes [0]_{B^{\mathrm{in}}}, ([1]_{B^{\mathrm{out}}})^T \otimes [1]_{B^{\mathrm{in}}} \}\; ,\\
        \mathcal{I}_{C} &= \Bigl\{ ([\Phi_+]_{C^{\mathrm{out}}})^T \otimes \mathbb{I}_{C^{\mathrm{in}}}\Bigr\}\; ,
    \end{align}
    where the output of $C$ factorises as $C^{\mathrm{out}} = C_A^{\mathrm{out}} \otimes C_B^{\mathrm{out}}$ and where $|\Phi_+\rangle = \frac{1}{\sqrt{2}}(|0\rangle_{C_A^{\mathrm{out}}}|0\rangle_{C_B^{\mathrm{out}}} + |1\rangle_{C_A^{\mathrm{out}}}|1\rangle_{C_B^{\mathrm{out}}})$ is a Bell state. Let the unitary channel $\rho_{AB|C}^U$ be given by identities
    \begin{equation}\label{rho_AB|C}
        \rho_{AB|C}^U = \rho_{A|C_A^{\mathrm{out}}}^{\mathrm{id}} \rho_{B|C_B^{\mathrm{out}}}^{\mathrm{id}}\;. 
    \end{equation}
    \red Here for simplicity we consider a case where we have complete information about the event at the common cause node $C$ (that is, the single-outcome instrument that prepares the Bell state), and thus there is no exogenous node, and no abduction step is necessary. Now consider the counterfactual query 
    \begin{quote}
    $Q_1$:\emph{``Given that $a = b = 0$, what is the probability that $b' = 1$ had it been that $a' = 1$?"}.
    \end{quote}
    
    This \red query is ambiguous until we specify what instruments those counterfactual outcomes correspond to. \pur On the one hand, interpreting as a do-interventional counterfactual with \red $\tau_A^{a'=1}=\tau_A^{\mathrm{do}([1])}$, while keeping the instrument at $B$ fixed -- the most parsimonious interpretation, since the consequent of a counterfactual does not call for an intervention even in the classical case -- \pur we obtain the (expected) counterfactual probability \red (which in this case are the same as there is no abduction involved),
\begin{equation}\label{eq:Bell_interventional_CFP}
        P_{\mathrm{do}(a'=1)}^{a=b=0}(b'=1) =P_{\mathrm{do}(a'=1)}(b'=1|a'=1)=\frac{1}{2} \;.
\end{equation}
Similarly, consider the query 
\begin{quote}
   $Q_2$:\emph{``Given that $a = b = 0$, what is the probability that $b' = 1$ had it been that $a' = 0$?"}.
\end{quote}
Interpreted do-interventionally, the answer to this query is
\begin{equation}\label{eq:Bell_interventional_CFP}
        P_{\mathrm{do}(a'=0)}^{a=b=0}(b'=1) =P_{\mathrm{do}(a'=1)}(b'=1|a'=1)=\frac{1}{2} \;.
\end{equation}
\red In other words, there is no counterfactual dependence between $b'$ and $a'$ in this case, when the counterfactual antecedent is interpreted as the outcome of a do-intervention.
\blk
    
\red
On the other hand, both $Q_1$ and $Q_2$ can be interpreted as a \emph{passive} counterfactual query (with $\bz'=\bz$), since the antecedent of both queries have nonzero model probabilities $P_{\bz'=\bz}(a'=1) = P_{\bz'=\bz}(a'=0)=\frac{1}{2}$. In this passive reading, we obtain
\begin{equation}\label{eq: Bell_passive_1}
    P^{a=b=0|\bz}_{a'=1|\bz'=\bz}(b'=1) = P_{\bz'=\bz}(b'=1|a'=1)=1\,
\end{equation}
for $Q_1$, and
\begin{equation}\label{eq: Bell_passive_2}
    P_{a'=0|\textbf{z'}=\bz}^{a=b=0|\bz}(b'=1) =  P_{\bz'=\bz}(b'=1|a'=0)=0 \;
\end{equation}
for $Q_2$. In other words, 
According to the above equations, it would have been the case that $b'=1$ with certainty had it been the case that $a'=1$, \emph{and} it would have been the case that $b'=0$ with certainty had it been the case that $a'=0$ (which in fact, is the actual case). In other words, the outcomes at $A$ and $B$ are counterfactually dependent. 
\blk
\end{example}

In Pearl's classical semantics, counterfactual dependence of the type in \red Eqs.~(\ref{eq: Bell_passive_1}) and \eqref{eq: Bell_passive_2} \blk would imply that $A$ is a cause of $B$.\footnote{In Lewis's account \cite{Lewis1973causation}, such counterfactual dependence also implies causal dependence. The difference is that Pearl analyzes counterfactuals in terms of causation, which he takes to be more fundamental, whereas Lewis analyzes causation in terms of counterfactuals, which he takes as more fundamental.} Nevertheless, the quantum structural causal model we used to derive this result has by construction no causal dependence from $A$ to $B$. This shows that \emph{in quantum causal models, (passive) counterfactual dependence does not imply causal dependence}. 

Note also that in the passive reading, the counterfactual antecedent corresponds to an event -- in the technical sense of an instrument element, that is, a CP map -- that was one of the potential outcomes of the actual instrument. A counterfactual antecedent interpreted as a do-intervention, on the other hand, is a \emph{different event} altogether -- technically distinct from any event in the actual instrument. This fact is obscured in the classical case, since in Pearl's formalism we identify the incoming and outgoing systems, and it is implicitly assumed that we can always (at least in principle) perform ideal non-disturbing measurements of the variables involved. Classically, the event `$X=x$' can ambiguously correspond to ``an ideal non-disturbing measurement of $X$ has produced the outcome $x$'' or ``the variable $X$ was set to the value $X=x$''. The distinction between those interpretations, in Pearl, is attributed to the structural relations in the model; the second interpretation is represented as a surgical excision of causal arrows, while leaving the variables themselves otherwise intact. In a quantum causal model, on the other hand, a do-intervention corresponds to a related but \emph{technically distinct} event in an otherwise intact model.

In our view, counterfactual dependence without causal dependence is a much more distinctively quantum feature than the failure of ``counterfactual definiteness'', which we turn to next.

\section{A Note on Counterfactual Definiteness and Bell's theorem}\label{sec: CFD}
As we've seen in Sec.~\ref{sec:analysis}, one of the features of our formalism is that a counterfactual proposition is not always either true or false, unlike in the classical semantics. In the classical case, a structural causal model underpins the deterministic nature of the system by defining functional dependencies among the nodes. This is not in the quantum case, even with full knowledge about a quantum structural model. Whereas in the classical case we can define the \emph{probability of a counterfactual} (to be true), in the quantum case we can in general only define a \emph{counterfactual probability}. 

The lack of definite truth values for counterfactuals in quantum mechanics can be thought of as a failure of ``counterfactual definiteness'', a concept that has a long and controversial history in discussions of Bell's theorem. Skyrms \cite{skyrms1982counterfactual} defines counterfactual definiteness (CFD) as follows, (attributing it to Stapp \cite{stapp1971s}, who expressed the idea that Bell's theorem requires it as an underlying assumption):
\begin{quote}
    ``Counterfactual definiteness; essentially the assumption that subjective conditionals of the form: `If measurement $M$ had been performed, result $R$ would have been obtained' always have a definite truth value (even for measurements that were not carried out because incompatible measurements were being made) and that the quantum mechanical statistics are the probabilities of such conditionals.''
\end{quote}
Skyrms \cite{skyrms1982counterfactual} argues, contrary to Stapp, that some forms of Bell's theorem can be proved using conditional probabilities rather than probabilities of subjunctive conditionals, and hence that CFD is not a necessary assumption in its derivation. Since then there has been a long debate about the status of CFD as an assumption underlying the derivation of Bell's theorem (see e.g.~\cite{blaylock2010epr, maudlin2010bell, griffiths2011epr, maudlin2011bell, lambare2021note, zukowski2014quantum}). Analysing the details of this long and nuanced literature is beyond our scope, but from the perspective of our framework, we can say some (hopefully) clarifying remarks. 

Firstly, Bell inequalities can be derived from many different sets of propositions~\cite{WisemanCavalcanti2017, Cavalcanti2021}. Disagreements often arise due to different ``camps'' using the same terms (most notably `locality') to refer to different concepts. Bell's 1964 theorem~\cite{Bell1964} explicitly used a notion of `locality' and an assumption of determinism (as well as an implicit "free choice" assumption). Bell inequalities can also be derived~\cite{WisemanCavalcanti2017}, however, by replacing Locality and determinism by Bell's 1976 notion of Local Causality~\cite{Bell1976} (stronger than Locality).

In the causal language we are using here, Bell's 1964 propositions, using determinism, are analogous to assuming the existence of a classical structural causal model (over the common-cause causal structure of a Bell scenario). As we've seen above, in Pearl's semantics, given a classical structural causal model, counterfactuals always have well-defined truth values. When one argues that CFD is \emph{necessary} for derivations of Bell inequalities, rather than `locality' alone, one is (we surmise) implicitly assuming something like Bell's 1964 notions of Locality, rather than Local Causality. To derive a Bell inequality, the notion of Locality indeed needs to be supplemented with something else, and the upshot is that CFD carries the same effect in this context as assuming determinism, as Bell did in 1964.

However, one may also assume Local Causality instead, as Bell did in 1976 (or indeed other assumptions~\cite{WisemanCavalcanti2017, Cavalcanti2021}), without assuming determinism. In this case, in causal language, one can effectively assume a Classical Causal Model, without further assuming that there exists an underlying Structural Causal Model. In this case, CFD may in general fail, or at least we fail to have the structure necessary to define Pearl's semantics of counterfactuals. Thus, Bell inequalities \emph{can} be derived without assuming CFD. 

Nevertheless, one may argue (essentially as in \cite{maudlin2010bell}), that determinism can be derived from Local Causality plus the perfect correlations of a pure entangled state. This would then make CFD true after all. However, perfect correlations are not observable in practice in real experiments, and in any case this does not change the fact that the derivation of Bell inequalities does not \emph{require} perfect correlations (and indeed this is what makes them experimentally testable!).

Clearly, CFD as defined by Skyrms \cite{skyrms1982counterfactual} does indeed fail in the quantum framework presented here. This however does not imply that Bell's theorem can be resolved \emph{merely} by rejecting CFD. Indeed, CFD may fail even in a completely \emph{classical} but \emph{indeterministic} causal model. The violation of Bell inequalities is explained within quantum causal models not simply as due to the indeterminism of quantum theory, but something deeper. One way of thinking about it is that it is due to the failure of Reichenbach's principle of common cause, which is in turn a consequence of the Classical Causal Markov Condition: in quantum causal models, a complete specification of the causes of an event (in the case of a Bell scenario, the preparation of the entangled state) does not in general render it uncorrelated with its non-effects -- even if those non-effects are space-like separated, like the outcome of a distant instrument. Another way of thinking about it, we argue here, is that instead of the failure of CFD, the ``quantumness'' of quantum causal models is better captured by the fact that, as discussed in Sec.~\ref{sec: Bell scenario}, quantum causal models allow for counterfactual dependence without causal dependence. A more detailed discussion of this point, and how it does not arise merely as an artefact of the split-node structure of quantum causal models, will be left for future work~\cite{KooderiSuresh2024}.

\section{Generalisations and related work: quantum Bayes' theorem}\label{sec: quantum Bayes}
Thm.~\ref{thm: quantum extension} and Cor.~\ref{cor: lifting classical counterfactuals} show that our formalism for counterfactuals in quantum causal models (see Sec.~\ref{sec: counterfactuals in QSM}) is a valid generalization of Pearl’s formalism in the classical case (see Sec.~\ref{sec: classical causal models}). In this section, we review the key assumptions of our formalism, discuss possible generalizations, and draw parallels with related work on quantum Bayesian inference.

Recall that our notion of a `quantum counterfactual’ in Def.~\ref{def: quantum counterfactual query} is evaluated with respect to a quantum structural causal model (QSM) (see Def.~\ref{def: quantum structural causal model}). A QSM $\Q$ reproduces a given physical process operator $\sigma_\bA$ over observed nodes $\bA$, that is, $\sigma_\bA$ arises from coarse-graining of ancillary (environmental) degrees of freedom in $\Q$ (cf. Eq.~(\ref{eq: compatibility with process operator})). As such, $\Q$ encodes additional information that is not present in $\sigma_\bA$: namely, (i) it assumes an underlying unitary process $\rho_{\bA S|\bA \bL}^U$, and (ii) it incorporates partial knowledge about the preparation of ancillary states at exogenous nodes in the form of preparation instruments $\{\tau^{\lambda_i}_{\Lambda_i}\}_{\lambda_i}$ (cf. Eq.~(\ref{eq: exogenous preparation instruments})), acting on an arbitrary input state. Together, this allowed us to reduce the abduction step in our formalism to classical Bayesian inference.

We remark that this situation (of a unitary background process with ancillas prepared in a fixed basis) arises naturally in the context of quantum circuits, future quantum computers, and thus supposedly in the context of future quantum AI. Nevertheless, for other use cases it might be less clear how to model our background knowledge on a physical process $\sigma_\bA$ in terms of a QSM, thus prompting relaxations of the assumptions baked into Def.~\ref{def: quantum structural causal model}. First, one may want to drop our assumption of a unitary background process. This assumption closely resembles Pearl’s classical formalism, which models any uncertainty about a stochastic physical process as a probabilistic mixture of deterministic processes. 
Yet, one might argue that assuming a unitary background process is too restrictive (or perhaps even in general fundamentally unwarranted) and that one should allow for arbitrary convex decompositions of a quantum stochastic process (CPTP map). 
To this end, note that knowledge about stable facts \cite{DiBiagioRovelli2021} that lead to a preferred convex decomposition of the process operator $\sigma_\bA = \sum_{\bm{\lambda}} P(\bm{\lambda}) \sigma^{\bm{\lambda}}_\bA$ (into valid process operators $\sigma^{\bm{\lambda}}_\bA$) is all that is necessary to perform (classical) Bayesian inference (cf.~Eq.~(\ref{eq: Bayesian update in QSM})).

A more radical generalization could arise by taking our information about exogenous variables to be inherently quantum. That is, without information in the form of stable facts regarding the distribution over outcomes of preparation instruments in a QSM, our knowledge about exogenous variables merely takes the form of a generic quantum state $\rho_{\bm{\Lambda}} = \rho_{\Lambda_1} \otimes \cdots \otimes \rho_{\Lambda_n}$.\footnote{Note that without the extra information about exogenous instruments in a QSM, we reduce to the situation described by a ``unitary process operator with inputs'', as defined in Def.~4.5 of Ref.~\cite{LorenzOreshkovBarrett2019} (see also Thm.~\ref{thm: compatibility and Markovianity}).} In this case, inference can no longer be described by (classical) Bayes’ theorem but requires a quantum generalization. 
Much recent work has been devoted to finding a generalization of Bayes’ theorem to the quantum case, which has given rise to various different proposals for a quantum Bayesian inverse -- 
see Ref.~\cite{ChoJacobs2019,Fritz2020,ParzygnatRusso2021} for a recent categorical (process-diagrammatic) definition and Ref.~\cite{ParzygnatBuscemi2022} for an attempt at an axiomatic derivation. 
Once a definition for the Bayesian inverse has been fixed - and provided it exists\footnote{Ref.~\cite{ParzygnatRusso2021} characterises the existence of a Bayesian inverse in the categorical setting for finite-dimensional $C^*$-algebras.}
- we can perform a generalized abduction step in Sec.~\ref{sec:cf_q_evaluation} and, consequently, obtain a generalised formalism for counterfactuals. We leave a more careful analysis of counterfactuals arising in this way and their comparison to our formalism for future study.

\section{Conclusion}\label{sec: conclusion}
We defined a notion of counterfactual in quantum causal models and provided a semantics to evaluate counterfactual probabilities, generalizing the three-step procedure of evaluating probabilities of counterfactuals in classical causal models due to Pearl \cite{Pearl2000}. The third level in Pearl's ladder of causality (see Fig.~\ref{fig:ladder}) had thus far remained an open gap in the generalization of Pearl's formalism to quantum causal models; here, we fill this gap.

To this end, we introduce the notion of a quantum structural causal model, which takes inspiration from Pearl's notion of a classical structural causal model, yet differs from the latter in several ways. A quantum structural model is fundamentally probabilistic; it does not assign truth values to all counterfactuals, and in this sense violates ``counterfactual definiteness''.

Despite these differences, we prove that every classical structural causal model admits an extension to a quantum structural causal model, which preserves the relevant independence conditions and yields the same probabilities for counterfactual queries arising in the classical case. Thus, quantum structural causal models and the evaluation of counterfactuals therein subsume Pearl's classical formalism.

On the other hand, quantum structural causal models have a richer structure than their classical counterparts. We identify different types of counterfactual queries arising in quantum causal models, and explain how they are distinguished from counterfactual queries in classical causal models. Based on this distinction, we evaluate these different types of quantum counterfactual queries in the Bell scenario and show that counterfactual dependence does not generally imply causal dependence in this case. In this way, our analysis provides a new way of understanding how quantum causal models generalize Reichenbach's principle of common cause to the quantum case \cite{Reichenbach1991,LeiferSpekkens2013,cavalcanti2014modifications}: a quantum common cause allows for counterfactual dependence without causal dependence, unlike a classical common cause.

Our work opens up several avenues for future study. Of practical importance are applications of counterfactuals in quantum technologies. For example, questions such as ``Given that certain outcomes were observed at receiver nodes in a quantum network, what is the probability that different outcomes would have been observed, had there been an eavesdropper in the network?" can be relevant for security applications. 

It is well-known that quantum theory violates ``counterfactual definiteness'' in the sense of the phenomenon often referred to as `quantum contextuality' \cite{KochenSpecker1967}. The latter has been identified as a key distinguishing feature between classical and quantum physics, as well as a resource for quantum computation \cite{Raussendorf2013,HowardEtAl2014,FrembsRobertsBartlett2018,FrembsRobertsCampbellBartlett2023}. It would thus be interesting to study contextuality from the perspective of the counterfactual semantics spelled out here.

Finally, our analysis hinges on the classicality (`stable facts' in Ref.~\cite{DiBiagioRovelli2021}) of background (exogenous) variables in the model, as it allows us to apply (classical) Bayes' inference on our classical knowledge about exogenous variables. In turn, considering the possibility of our `prior' knowledge about exogenous variables to be genuinely quantum motivates a generalization of Bayes' theorem to the quantum case (see Sec.~\ref{sec: quantum Bayes}). We expect that combining our ideas with recent progress along those lines will constitute a fruitful direction for future research.

\paragraph*{Acknowledgements.} The authors acknowledge financial support through grant number FQXi-RFP-1807 from the Foundational Questions Institute and Fetzer Franklin Fund, a donor advised fund of Silicon Valley Community Foundation, and ARC Future Fellowship FT180100317.

\bibliographystyle{unsrt}
\bibliography{bibliography}

\clearpage

\appendix

\section{Proof of Thm.~\ref{thm: quantum extension}}\label{app: proof quantum extension}

The proof consists of several parts: (i) we find a binary extension of a classical structural causal model (CSM), (ii) we provide a protocol that extends a (binary) CSM to a reversible one, where all functional relations are bijective, (iii) we encode classical copy operations in a quantum structural causal model (QSM) using CNOT-gates. In the final part (iv), we combine (i)-(iii) to construct a QSM $\Q$, which (linearly) extends a PSM $\langle \M,P(\bu)\rangle$ as desired.\\

\textbf{(i) binary encoding: every CSM has a binary extension.} Let $\M = \langle \bU,\bV,\bF\rangle$ be a CSM (see Def.~\ref{def: classical_structural_cm_def}). First, we enlarge the sets $\bV,\bU$ in $\M$ to sets of cardinality a power of two. For all $i\in\{1,\cdots,n\}$, let $V^{(b)}_i$ denote a set of cardinality $N_i = |V^{(b)}_i| = \lceil\log_2|V_i|\rceil$. 
Second, we extend $\bF$ to the enlarged sets $U^{(b)}_i$ and $V^{(b)}_i$ as follows. For every $\bF \ni f_i: U_i \times \Pa_i \ra V_i$, let $f^{(b)}_i: U^{(b)}_i \times \Pa^{(b)}_i \ra V_i^{(b)}$ be 
the function given by
\begin{equation}\label{eq: binary encoding}
    f^{(b)}_i(u_i,p_1,\cdots,p_{J_i}) = 
    \begin{cases}
        f_i(u_i,p_1,\cdots,p_{J_i}) &\text{whenever\ } u_i \in U_i, p_j\in \Pa_i\; , \\
        f_i(*,*) &\text{otherwise\; .} 
    \end{cases}
\end{equation}
In words, $f_i^{(b)}$ identifies all elements in $U^{(b)}_i\backslash U_i$ and all elements in $\Pa^{(b)}_i \backslash \Pa_i$ with the same (arbitrary) element $(*,*) \in U_i \times \Pa_i$. It is easy to see that the causal structure of the CSM $\M^{(b)} = \langle \bU^{(b)},\bV^{(b)},\bF^{(b)} \rangle$ is the same as that of $\M$: the extensions $U_i^{(b)}$ and $V_i^{(b)}$ of $U_i$ and $V_i$ are defined locally, and $f^{(b)}_i$ has the same functional form as $f_i$. Consequently, $\M^{(b)}$ admits a (Markov) factorization of the same form as $\M$ in Eq.~(\ref{eq: causal Markov condition for M}).\\

\textbf{(ii) reversibility: every (binary) CSM has a (binary) reversible extension.} Let $\M = \langle \bU,\bV,\bF\rangle$ be a CSM. For every $f_i \in \bF$, we decompose its domain according to its pre-images $f_i^{-1}(v_i)$ for all $v_i \in V_i$. In particular, let $\tilde{S}_i$ be a set of cardinality $|\tilde{S}_i| := \max_{v_i \in V_i} \{|f_i^{-1}(v_i)|\}$. Next, we label the elements in $f^{-1}_i(v_i)$. We will use the lexicographic order $l_i: U_i \times \Pa_i \ra \{1,\cdots,|U_i||\Pa_i|\}$ on (any choice of) alphabets labeling the elements in the sets $U_i$ and $\Pa_i$, respectively. For every $v_i \in V_i$, let $l_{v_i}: f^{-1}_i(v_i) \ra \{1,\cdots,|f^{-1}_i(v_i)|\}$ be any order on $f^{-1}_i(v_i)$ such that $l_{v_i}(x) \leq l_{v_i}(y):\Leftrightarrow l_i(x) \leq l_i(y)$ for all $x,y \in f^{-1}_i(v_i)$. Combining orders $l_{v_i}$ on $f_i^{-1}(v_i)$, we obtain an order on $\bigcupdot_{v_i\in V_i} f^{-1}_i(v_i) = U_i \times \Pa_i$, denoted by $l_{f_i}: U_i \times \Pa_i \ra \{1,\cdots,|V_i||\tilde{S}_i|\}$. With this order, we define an extension $\tilde{f}_i: U_i \times \Pa_i \ra V_i \times \tilde{S}_i$ of the form
\begin{equation*}\label{eq: extension I}
    \tilde{f}_i(u_i,p_1,\cdots,p_{J_i}) \, .
\end{equation*}
Let $\pi_{V_i}: V_i \times \tilde{S}_i \ra V_i$ denote the projection onto the first factor $V_i$, i.e., $\pi_{V_i}(v_i,s_i) := v_i$; clearly, $f_i \cong \pi_{V_i} \circ \tilde{f}_i$.\footnote{$\pi_{V_i}$ can be interpreted as a `discarding operation'.} 

Now, $\tilde{f}_i$ is injective, however, it is not surjective in general. In order to obtain a bijective map, we also need to extend the domain of $f_i$. Let $S'_i$, $T'_i$ and $U'_i$ be sets of cardinalities such that
\begin{equation}
	|U'_i||\Pa_i| = |T'_i| |U_i| |\Pa_i| = |V_i| |S'_i|\; .
\end{equation}
Similar to before, we denote the lexicographic order corresponding to (some choice of) alphabets on $U'_i := T'_i \times U_i$ and $\Pa_i$ by $l_i$. Let $\bigcupdot_{v_i \in V_i} S'_i(v_i) = U'_i \times \Pa_i$ be any decomposition such that $f^{-1}_i(v_i) \subseteq S'_i(v_i)$ and $|S'_i(v_i)| = |S'_i|$ for all $v_i \in V_i$. Moreover, let $l_{f_i}: U'_i \times \Pa_i \ra \{1,\cdots,|V_i||S'_i|\}$ be any order such that $l_{f_i}(x) \leq l_{f_i}(y) \Leftrightarrow l_i(x) \leq l_i(y)$ for all $x,y \in S'_i(v_i)$ and $v_i \in V_i$. Then we define $f'_i: U'_i \times \Pa_i \ra V_i \times S'_i$ by
\begin{equation}\label{eq: reversible function extension}
	f'_i(u'_i,p_1,\cdots,p_{J_i}) := (f_i(u_i,p_1,\cdots,p_{J_i}), l_{f_i}(u'_i,p_1,\cdots,p_{J_i}))\; .
\end{equation}
In this case, we have $f_i \circ \pi_{U_i} = \pi_{V_i} \circ f'_i$. Moreover, $f'_i$ is a bijection by construction.

Finally, we note that the extensions $f_i \ra f'_i$ are local: they only add degrees of freedom $T'_i$ and $S'_i$, which are local to the node $V_i$, $T'_i$ is a cause of $S'_i$ only and the $S'_i$ are all discarded. It follows that---up to an additional `sink' node $S' = \times_{i=1}^n S'_i$---the CSM $\M' = \langle \bU',\bV',\bF'\rangle$ for $\bU' = \{T'_1\times U_1,\cdots,T'_n\times U_n\}$, $\bV' = \{\bV,S'\}$ and $\bF' = \{f'_1,\cdots,f'_n\}$ has the same causal structure as $\M$, that is, $\M'$ admits a (Markov) factorisation of the same form as $\M$ in Eq.~(\ref{eq: causal Markov condition for M}). Moreover, note that we can apply the above protocol to the binary extension $\M^{(b)}$ of a CSM $\M$, resulting in a binary, reversible CSM, which we will denote by $\M'$ also.\footnote{Note, that the binary encoding in Eq.~(\ref{eq: binary encoding}) will generally increase the cardinality of the sink nodes $S'_i$.}\\

\textbf{(iii) copy operation.} Since classical information can be copied and processed to various nodes in a CSM $\M$, we need a way to encode such copying in a quantum extension $\Q$. Indeed, classical information is copied in the CSM $\M = \langle \bU,\bV,\bF \rangle$ whenever there exist $f_j, f_{j'} \in \bF$, $j\neq j'$ such that $f_j^{-1}(V_j) \cap f_{j'}^{-1}(V_{j'}) \neq \emptyset$. Let $\mathrm{Ch}(V_i) := \{V_j \in \bV \mid V_i \in f^{-1}_j(V_j)\} \subset \bV$ be the set containing all nodes $V_j \in \bV$ that are related to $V_i$ via $\bF$. 
Since we are only interested in proving the existence of a quantum extension $\Q$ of $\M$, we will content ourselves with an inefficient copy protocol: namely, we will copy all of $V_i$ to every node $V_j \in \mathrm{Ch}(V_i)$.\footnote{This copy protocol will generally increase the cardinality of variables $\bV$ in the CSM, and is inefficient in this regard.}

The idea is to use CNOT gates to implement classical copying in a fixed basis. This requires additional ancillary (quantum) nodes in the quantum extension $\Q$ of a CSM $\M$, which we will denote by $\bm{\Lambda}'$. More precisely, assume that $V$ are quantum nodes with associated Hilbert spaces $\cH_{V^{\mathrm{in}}_i}, \cH_{V^{\mathrm{out}}_i}$ of respective dimension $\mathrm{dim}(\cH_{V^{\mathrm{in}}_i}) = \mathrm{dim}(\cH_{V^{\mathrm{out}}_i}) = 2^{N_i}$.\footnote{By (i) and (ii), we may assume that $\M$ is a binary, reversible CSM, in particular, $|V_i| = 2^{N_i}$.} Further, let $\Lambda'_i := \otimes_{V_j \in \mathrm{Ch}(V_i)} \Lambda'_{ji}$, where $\Lambda'_{ji}$ denotes a quantum node with associated Hilbert spaces $\cH_{\Lambda^{\mathrm{in}}_{ji}}, \cH_{\Lambda^{\mathrm{out}}_{ji}}$  of the same dimension as $V_i$, that is, $\mathrm{dim}(\cH_{\Lambda_{ji}'^{\mathrm{out}}}) = \mathrm{dim}(\cH_{\Lambda_{ji}'^{\mathrm{out}}}) = 2^{N_i}$. Then we devise a copy protocol involving a total of $|\mathrm{Ch}(V_i)|N_i$ CNOT-gates as follows,\footnote{Here and below, we implicitly assume the individual CNOT gates in $C_i$ to be `padded' with identities on ancillary systems not involved in the present CNOT gate. As an example, see the copy operation between three-qubit nodes in Eq.~(\ref{eq: two qubit copy operation}). \label{fn: padding}} 
\begin{align}\label{eq: copy operation}
    C_i := \prod_{V_j\in \mathrm{Ch}(V_i)} \left(\prod_{k=1}^{N_i} \mathrm{CNOT}_{\Lambda'_{ji,k}V_{i,k}}\right)\; ,
\end{align}
where $\mathrm{CNOT}_{\Lambda'_{ji,k}V_{i,k}}$ denotes the CNOT gate with control and target given by the $k-th$ qubit in $V_i$ and $\Lambda'_{ji}$, respectively; that is,
\begin{equation*}\label{eq: CNOT gate}
    \mathrm{CNOT}_{\Lambda'_{ji,k}V_{i,k}}(|0\rangle \otimes |q\rangle) = |q\rangle \otimes |q\rangle\; ,
\end{equation*}
whenever $|q\rangle \in \{|0\rangle,|1\rangle\}$. We will adopt the convention that the Hilbert space $\mc{H}_{V_i^{\mathrm{in}}}$ after the operation $C_i$ is discarded. Note that Eq.~(\ref{eq: copy operation}) is then valid for $\mathrm{Ch}(V_i) = \emptyset$, namely we have $C_i = \mathbb{I}_{V_i^\mathrm{in}}$ in this case.

We are left to show that the encoding in Eq.~(\ref{eq: copy operation}) satisfies the relations $\lbrace \Lambda'_j \nrightarrow V_i\rbrace$ whenever $j\neq i$ in Def.~\ref{def: quantum structural causal model} for the additional ancillary nodes $\Lambda'_j$. This follows immediately from the commutation relations of CNOT gates (see Prop.~1 in \cite{bataille2020quantum}), together with Thm.~\ref{thm: compatibility and Markovianity}.

Below, we explicitly analyze the case of two CNOT-gates, copying a single qubit from a node $A$ to two children $B,C \in \mathrm{Ch}(A)$. The situation is depicted in Fig.~\ref{fig: two CNOT gates}.

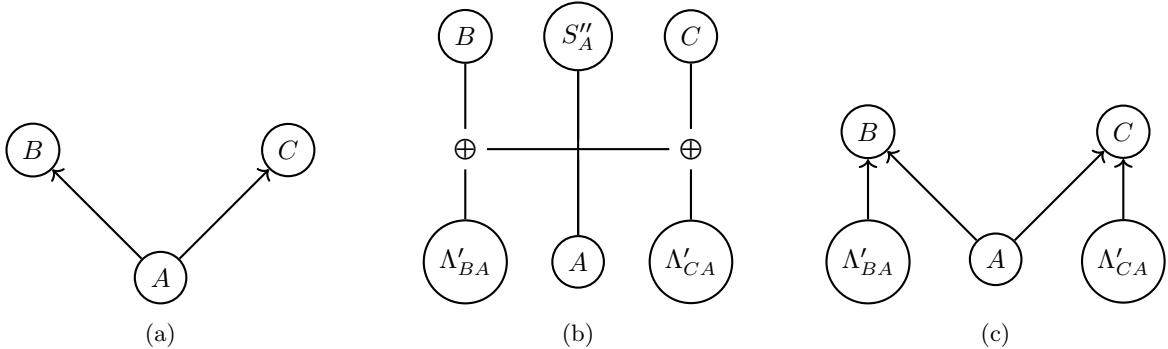
\begin{figure}[!htp]
    \begin{minipage}[b]{0.33\textwidth}
	\centering
	\begin{tikzpicture}[node distance={24mm}, thick, main/.style = {draw, circle}]
		\node[main] (2) [ above right of=1 ] {$B$}; 
		\node[main] (3) [ below right of=2 ] {$A$};
		\node[main] (4) [ above right of=3 ] {$C$}; 
		
            \draw[->] (3) -- (2);
 	      \draw[->] (3) -- (4);
        \end{tikzpicture}
 	\subcaption{}
    \end{minipage}
    \begin{minipage}[b]{0.33\textwidth}
	\centering
	\begin{tikzpicture}[node distance={15mm}, thick, main/.style = {draw, circle}] 
		\node[main] (1) {$\Lambda'_{BA}$}; 
		\node[main] (2) [ right of=1 ] {$A$}; 
		\node[main] (3) [ right of=2 ] {$\Lambda'_{CA}$};
		\node[] (4) [ above of=1 ] {$\bigoplus$}; 
		\node[] (5) [ above of=3 ] {$\bigoplus$};
		\node[] (6) [ above of=2 ] {}; 
		\node[main] (7) [ above of=4 ] {$B$};
		\node[main] (8) [ above of=6 ] {$S''_A$};
		\node[main] (9) [ above of=5 ] {$C$};
		  
            \draw[-] (1) -- (4); 
 		\draw[-] (4) -- (7);
 		\draw[-] (2) -- (6);
 		\draw[-] (3) -- (5);
 		\draw[-] (6) -- (8); 
 		\draw[-] (5) -- (9);
 		\draw[-] (2) -- (8);
 		\draw[-] (4) -- (5);
 	\end{tikzpicture}
        \subcaption{}
    \end{minipage}
    \begin{minipage}[b]{0.33\textwidth}
	\centering
	\begin{tikzpicture}[node distance={24mm}, thick, main/.style = {draw, circle}] 
		\node[main] (A) {$A$}; 	
                \node[main] (B) [ above left of=A] {$B$}; 
                \node[main] (C) [ above right of=A] {$C$}; 
                \node[main] [below=0.8cm of B] (LAB) {$\Lambda'_{BA}$};
			\node[main] (LAC) [below=0.8cm of C] {$\Lambda'_{CA}$};
			\draw[->] (A) -- (B); 
			\draw[->] (A) -- (C);
			\draw[->] (LAB) -- (B); 
			\draw[->] (LAC) -- (C);
 	\end{tikzpicture}
        \subcaption{}
    \end{minipage}
    \caption{Example of copy protocol as described in (iii): (a) a quantum causal model involving three single-qubit quantum nodes $\bA = \{A,B,C\}$, (b) the qubit at node $A$ is copied onto ancillary systems $\Lambda'_{BA}$ and $\Lambda'_{CA}$, using two CNOT gates $\mathrm{CNOT}_{\Lambda'_{BA}A}$ and $\mathrm{CNOT}_{\Lambda'_{CA}A}$, (c) the copy ancillae $\Lambda'_{BA}$ and $\Lambda'_{CA}$ enter as additional endogenous variables in the resulting causal model.}
    \label{fig: two CNOT gates}
\end{figure}
We want to show that $\Lambda'_{BA}$ and $\Lambda'_{CA}$ are local influences only, that is, $\Lambda'_{BA}$ does not signal to $C$ and $\Lambda'_{CA}$ does not signal to $B$.
Let 
\begin{gather}
    \rho_{\Lambda'_{BA}} = 
    \begin{bmatrix}
    \alpha & \beta\\
    \beta^{\ast} & 1-\alpha
    \end{bmatrix}
\end{gather}
be a state at $\Lambda'_{BA}$ and let
\begin{gather}
    \rho_{A\Lambda'_{CA}} = 
    \begin{bmatrix}
    a & b & c & d\\
    b^{\ast} & f & g & h\\
    c^{\ast} & g^{\ast} & k & l\\
    d^{\ast} & h^{\ast} & l^{\ast} & 1-a-f-k
    \end{bmatrix}
\end{gather}
be any bipartite state over the joint system $A$ and $\Lambda'_{CA}$.\footnote{We remark that in the QSM $\Q$ constructed below, $\rho_{A\Lambda'_{CA}}$ is in fact always a product state.} In order to check the no-influence relation `$\Lambda'_{BA} \ra C$' we evaluate
\begin{equation}\label{rho_d}
    \rho_{C} = \mathrm{Tr}_{AB}[U(\rho_{\Lambda'_{BA}} \otimes \rho_{A\Lambda'_{CA}}) U^\dagger]\; ,
\end{equation}
where $U=C_A$ is the copy operation in Eq.~(\ref{eq: copy operation}), in our case, it is the product of two CNOT gates,
\begin{equation}\label{eq: two qubit copy operation}
     C_A = (\mathrm{CNOT}_{\Lambda'_{BA}A}\otimes \mathbb{I}_{\Lambda'_{CA}})
     (\mathbb{I}_{\Lambda'_{BA}} \otimes \mathrm{CNOT}_{\Lambda'_{CA}A})\; ,
\end{equation}
and is given by,
\begin{gather}
C_A =
    \begin{bmatrix}
1 & 0 & 0 & 0 & 0 & 0 & 0 & 0\\
0 & 1 & 0 & 0 & 0 & 0 & 0 & 0\\
0 & 0 & 0 & 0 & 0 & 0 & 0 & 1\\
0 & 0 & 0 & 0 & 0 & 0 & 1 & 0\\
0 & 0 & 0 & 0 & 1 & 0 & 0 & 0\\
0 & 0 & 0 & 0 & 0 & 1 & 0 & 0\\
0 & 0 & 0 & 1 & 0 & 0 & 0 & 0\\
0 & 0 & 1 & 0 & 0 & 0 & 0 & 0
\end{bmatrix}\; .
\end{gather}
By straightforward computation we obtain,
\begin{gather}
    \rho_C = 
    \begin{bmatrix}
    a+k & b^{\ast}+l^{\ast}\\
    b+l & 1-a-k
    \end{bmatrix}\; ,
\end{gather}
which shows that $\rho_C$ is independent of $\rho_{\Lambda'_{BA}}$. Similarly, one shows that $\rho_B$ is not affected by the matrix elements of $\rho_{\Lambda'_{CA}}$.
The argument can be extended to more than two (blocks of) CNOT gates at a node $V_i$ (cf. Eq.~(\ref{eq: copy operation})) and for nodes with more than one child node.\\

\textbf{(iv) quantum extension: every PSM has a quantum extension.} Let $\langle \M',P(\bu')\rangle$ be a PSM (see Def.~\ref{def: prob_cscm}) with CSM $\M' = \langle \bU',\bV',\bF'\rangle$ and $P(\bu')$ a probability distribution over exogenous variables of $\M'$. The structural relations $\bF'$ define independence conditions between variables $\bV'$. Representing these by a DAG, let $(V'_1,\cdots,V'_n)$ be an order compatible with the order of nodes in the DAG, that is, $V'_j \in Pa'(V'_i) \Rightarrow j<i$. Compatibility of $\M'$ with the DAG, equivalently the causal Markov condition in Def.~\ref{def: cmc}, reads
\begin{equation*}
    P(\bv')
    = \prod_{i=1}^n P(v'_i\mid pa'_i)
    = \sum_{u'_i}\prod_{i=1}^n \delta_{v'_i,f'_i(pa'_i,u'_i)} P(u'_i)\; .
\end{equation*}

By steps (i) and (ii), we may assume that $U'_i,V'_i,S'_i$ are sets of cardinality a power of two and that the $f'_i: U'_i \times Pa'_i \ra V'_i \times S'_i$ are bijective functions. In order to construct a QSM from $\langle \M',P(\bu')\rangle$, we promote these classical variables to quantum nodes, by associating them with the (free) Hilbert spaces $\cH_{U_i'^{\mathrm{in}}}$, $\cH_{U_i'^{\mathrm{out}}}$, $\cH_{V_i'^{\mathrm{in}}}$, $\cH_{V_i'^{\mathrm{out}}}$, and $\cH_{S_i'^{\mathrm{in}}}$, $\cH_{S_i'^{\mathrm{out}}}$ over the outcome sets associated to $U'_i,V'_i,S'_i$, together with a respective choice and orthonormal basis, as well as an identification of bases between input and output spaces, e.g. $\{|v'_i\rangle\}_{v'_i}$ in $\cH_{V_i'^{\mathrm{in}}} \cong \cH_{V_i'^{\mathrm{out}}}$. Moreover, since the $f'_i$ are bijective, complex-linear extension yields isometries
\begin{equation}
    \widetilde{W}_i : \cH_{U_i'^{\mathrm{out}}}\otimes\cH_{\Pa_i'^{\mathrm{out}}} \rightarrow \cH_{V_i'^{\mathrm{in}}}\otimes\cH_{S_i'^{\mathrm{in}}}\; .
\end{equation}
Part (iii) further yields unitaries implementing local classical copy operations,
\begin{align}
    C_i: \left(\bigotimes_{V'_j \in \mathrm{Ch}'(V'_i)} \mathcal{H}_{\Lambda_{ji}'^{\mathrm{out}}} \right) \otimes \mathcal{H}_{V_i'^{\mathrm{in}}} \longrightarrow \left(\bigotimes_{V'_j\in\mathrm{Ch}'(V'_i)}\mathcal{H}_{\Lambda_{ji}'^{\mathrm{out}}}\right) \otimes \mc{H}_{V_i'^{\mathrm{in}}}\; ,
\end{align}
where by our above convention, $\mc{H}_{S_i''^{\mathrm{in}}} := \mc{H}_{V_i'^{\mathrm{in}}}$. Setting $W_i := C_i\widetilde{W}_i$ (see Fig.~\ref{fig: local unitaries with copy}) and composing the $W_i$'s for all nodes $V'_i \in V'$ in order, we define a global isometry by$^{\ref{fn: padding}}$
\begin{equation}\label{eq: unitary extension}
    W = \prod_{i=1}^n W_i = \prod_{i=1}^n C_i\widetilde{W}_i\; .
\end{equation}
Recall that by our convention that $\mc{H}_{V_i'^{\mathrm{in}}}$ is discarded after the copy process, we have $C_i = \mathbb{I}_{V_i'^{\mathrm{in}}}$ for $\mathrm{Ch}'(V'_i) = \emptyset$. Moreover, adopting the notation of the copy protocol from Eq.~(\ref{eq: copy operation}), we set $Pa''_i := \otimes_{V'_j \in Pa'_i} \Lambda'_{ij}$ as well as $V''_i := \otimes_{V'_j \in \mathrm{Ch}'(V'_i)} \Lambda'_{ji}$ for all $V'_i$ with $\Pa'(V'_i) \neq \emptyset$ in the input and output quantum nodes of the isometry $\widetilde{W}_i: \cH_{U'^{\mathrm{out}}_i} \otimes \cH_{Pa''^{\mathrm{out}}_i} \ra \cH_{V''^{\mathrm{in}}_i} \otimes \cH_{S'^{\mathrm{in}}_i}$ and consequently for $W_i: \cH_{\Lambda''^{\mathrm{out}}_i} \otimes \cH_{Pa''^{\mathrm{out}}_i} \ra \cH_{V''^{\mathrm{in}}_i} \otimes \cH_{S''^{\mathrm{in}}_i}$, where $\Lambda''_i$ consists of exogenous variables $U'_i$ and $\Lambda'_i$, and $S''_i$ consists of $S'_i$ and $V'_i$ (see Fig.~\ref{fig: local unitaries with copy}).

Note that the ordering of unitaries in Eq.~(\ref{eq: unitary extension}) maintains the partial ordering of the directed acyclic graph corresponding to $\M'$. Consequently, $\rho^W_{\bV''S''|\bV''\bm{\Lambda}''}$ admits the factorisation in Eq.~(\ref{eq: factorisation of \Q}).
\begin{figure}[!htp]
    \begin{minipage}[t]{0.6\textwidth}
        \centering
        $\left.
        \begin{array}{r}
	\begin{tikzpicture}[]
	   \tikzstyle{operator1} = [draw, rectangle, minimum width=4cm, minimum height=1cm]
          \tikzstyle{operator2} = [draw, rectangle, minimum width=6cm, minimum height=1cm]
	   \tikzstyle{line} = [thick,-]
	   \node[operator2] at (1,0) (W1) {$\widetilde{W}_i$};
	   \node[operator1] at (0,3) (C1) {$C_i$};
		
	   \draw[line] (-1,-1) -- (-1,-.5);
	   \draw[line] (3,-1) -- (3,-.5);
	   \draw[line] (-1,.5) -- (-1,2.5);
	   \draw[line] (3,.5) -- (3,4);
			
	   \draw[line] (-1,3.5) -- (-1,4);
	   \draw[line] (1,3.5) -- (1,4);
          \draw[line] (1,2) -- (1,2.5);
			
	   \node[] at (-1, -1.5) {$U'_i$};
          \node[] at (-1.5, -2.25) {$U_i$};
	   \node[] at (-0.5, -2.25) {$T_i'$};
          \draw[line] (-1.4,-1.9) -- (-1.2,-1.7);
	   \draw[line] (-.6,-1.9) -- (-0.8,-1.7);
    
	   \node[] at (3, -1.5) {$\Pa''_i := \! \! \! \underset{V'_j \in \Pa'(V'_i)}{\bigotimes}\! \Lambda'_{ij}$};
	   \node[] at (-1.4, 1.5) {$V_i'^{\mathrm{in}}$};
	   \node[] at (3, 4.5) {$S'_i$};
	   \node[] at (-1, 4.5) {$V_i'^{\mathrm{in}}$};
	   \node[] at (1.3, 4.8) {$\overbrace{\Lambda'_i = \! \! \! \underset{V'_j \in \mathrm{Ch}'(V'_i)}{\bigotimes}\! \Lambda'_{ji}}^{V''_i}$};
        \node[] at (1.3, 1.5) {$\Lambda'_i = \! \! \! \underset{V'_j \in \mathrm{Ch}'(V'_i)}{\bigotimes}\! \Lambda'_{ji}$};
    \end{tikzpicture}
    \end{array}
    \quad \quad \quad \quad \right \}$
    \end{minipage}
    \begin{minipage}[t]{0.4\textwidth}
        \centering
        \vspace{-2.5cm}
	\begin{tikzpicture}[]
	    \tikzstyle{operator} = [draw, rectangle, minimum width=4cm, minimum height=1cm]
		\tikzstyle{line} = [thick,-]
		\node[operator] at (0,0) (W1) {$W_i$};
		\draw[line] (-1,-1) -- (-1,-.5);
		\draw[line] (1,-1) -- (1,-.5);
		\draw[line] (-1,.5) -- (-1,1);
		\draw[line] (1,.5) -- (1,1);
			
		\node[] at (-1, -1.5) {$\Lambda''_i$};
		\node[] at (-1.5, -2.25) {$U'_i$};
		\node[] at (-.5, -2.25) {$\Lambda'_i$};
            \draw[line] (-1.2,-1.75) -- (-1.4,-1.95);
		\draw[line] (-0.9,-1.75) -- (-0.7,-1.95);
  
		\node[] at (1, -1.5) {$\Pa''_i$}; 
		\node[] at (-1, 1.5) {$V''_i$}; 
		\node[] at (1, 1.5) {$S''_i$};
            \node[] at (0.5, 2.25) {$V'^{\mathrm{in}}_i$};
            \node[] at (1.5, 2.25) {$S'_i$};
            \draw[line] (0.6,2) -- (0.8,1.8);
		\draw[line] (1.4,2) -- (1.2,1.8);
        \end{tikzpicture}
    \end{minipage}
    \caption{By Eq.~(\ref{eq: unitary extension}), the global isometry $W$ in the quantum extension $\Q$ of a PSM $\langle \M,P(\bu)\rangle$ factorises---according to the causal Markov condition for $\M$---into local isometries $W_i: \cH_{Pa''_i} \otimes \cH_{\Lambda''_i} \ra \cH_{V''_i} \otimes \cH_{S''_i}$ for every quantum node $V''_i$, of the form on the right. Every local isometry $W_i$ further decomposes as an isometry $\widetilde{W}_i$, given as the complex linear extension of the bijective functions in the reversible, binary CSM $\M'$ (see steps (i) and (ii)), followed by a unitary $C_i$ copying the node $V'_i$ to all its children nodes (see step (iii)). In terms of the copy ancillae $\Lambda'_{ji}$, we therefore have $V''_i := \bigotimes_{V'_j \in \mathrm{Ch}'(V'_i)} \Lambda'_{ji}$ and $\Pa''_i := \bigotimes_{V'_j \in \Pa'(V'_i)} \Lambda'_{ij}$.}
    \label{fig: local unitaries with copy}
\end{figure}
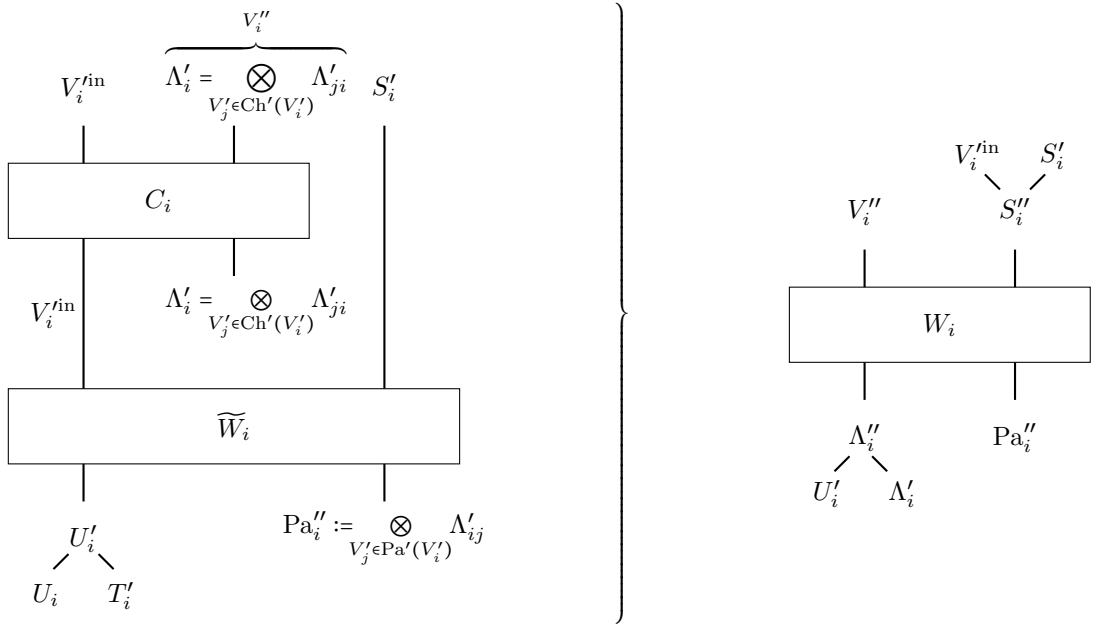

Finally, we encode the distribution $P(\bu)$ over exogenous nodes $U_i \in \bU$ of $\M$. First, we trivially extend $P(\bu)$ to $P^{(b)}(\bu^{(b)})$ by setting $P^{(b)}(\bu) = P(\bu)$ for $\bu \in \bU$ and $P^{(b)}(\bu) = 0$ for $\bu \in \bU^{(b)}\backslash \bU$ in step (i). Second, we define $P(\bu')$ as the product distribution of $P^{(b)}(\bu^{(b)})$ and the uniform distribution over $\bT' = \times_{i=1}^n T'_i$, i.e., $P(\bu') = P(\bu,\bt') = \frac{1}{|\bT'|} P^{(b)}(\bu)$ in step (ii). Finally, step (iii) requires us to initialise the input state of the copy ancillae $\Lambda'_i$. Taken together, we define instruments
\begin{equation}\label{eq: encoding of distributions over exogenous nodes in M}
    \{\tau^{u_i}_{\Lambda''_i}\}_{u_i} = \left\{\left(|0\rangle\langle 0|_{\Lambda'_i} \otimes P(u^{(b)}_i) |u^{(b)}_i\rangle\langle u^{(b)}_i|_{U^{(b)}_i} \otimes \frac{1}{|T'_i|}|t'_i\rangle\langle t'_i|_{T'_i}\right)^T \otimes \mathbb{I}_{\Lambda_i^{in}}\right\}_{u_i}\; .
\end{equation}
In summary, we thus obtain a QSM $\Q = \langle (\bV'',\bm{\Lambda}'',S''), \rho_{\bV''S''\mid \bV''\bm{\Lambda}''}^W, (\{\tau^{u_1}_{\Lambda''_1}\}_{u_1},\cdots,\{\tau^{u_n}_{\Lambda''_n}\}_{u_n})\rangle$, which by construction reproduces the classical probability distribution defined by $\M$ and $P(\bu)$ in Eq.~(\ref{eq: causal Markov condition for M II}), when evaluated on instruments corresponding to projective measurements in the ($|\mathrm{Ch}(V_i)|$-times copied) preferred bases, $\{\tau_{V''_i}^{v_i}\}_{v_i} = \{|v_i\rangle\langle v_i| \otimes |v_i\rangle\langle v_i|\}_{v_i}$. This completes the proof.

\end{document}